\documentclass{cc}
\usepackage{amssymb, amsmath,amsthm,mathtools}
\usepackage{xcolor}
\usepackage{graphicx}
\usepackage[utf8]{inputenc}
\usepackage[shortlabels]{enumitem}
\usepackage{hyperref}
\usepackage{pgf,tikz}
\setlist[itemize,enumerate]{itemsep=0pt,topsep=0pt}
\usetikzlibrary{arrows,backgrounds,positioning,fit,shapes,calc}

\tikzset{vertex/.style={circle,fill,draw,inner sep=0,minimum size=5pt}}
\DeclareMathOperator{\size}{size}
\DeclareMathOperator{\poly}{poly}
\newcommand{\eps}{\varepsilon}
\renewcommand{\epsilon}{\varepsilon}
\renewcommand{\bar}{\overline}

\newcommand{\N}{\mathbb N}
\newcommand{\Z}{\mathbb Z}
\newcommand{\Q}{\mathbb Q}
\newcommand{\R}{\mathbb R}
\newcommand{\C}{\mathbb C}

\contact{guusregts@gmail.com}% Insert contact email address here.
\submitted{13.12.2022}

\title{Approximating the chromatic polynomial is as hard as computing it exactly}
\titlehead{Approximating the chromatic polynomial is hard}
\author{Ferenc Bencs\\
Korteweg de Vries Institute\\ for Mathematics,\\
University of Amsterdam. \\ P.O. Box 94248 1090 GE
\\
Amsterdam, the Netherlands
\\
\email{ferenc.bencs@gmail.com}
\and
Jeroen Huijben
\\
Korteweg de Vries  Institute \\ for Mathematics,
\\
University of Amsterdam. \\  P.O. Box 94248 1090 GE 
\\
Amsterdam, the Netherlands
\\
\email{jeroenhuijben95@gmail.com} 
\and 
Guus Regts
\\
Korteweg de Vries Institute \\for Mathematics,
\\
University of Amsterdam. \\ P.O. Box 94248 1090 GE 
\\
Amsterdam, the Netherlands
\\
\email{guusregts@gmail.com}
}

\begin{abstract}
We show that for any non-real algebraic number $q$ such that $|q-1|>1$ or $\Re(q)>\frac{3}{2}$ it is \textsc{\#P}-hard to compute a multiplicative (resp. additive) approximation to the absolute value (resp. argument) of the chromatic polynomial evaluated at $q$ on planar graphs. 
This implies \textsc{\#P}-hardness for all non-real algebraic $q$ on the family of all graphs.
We moreover prove several hardness results for $q$ such that $|q-1|\leq 1$.

Our hardness results are obtained by showing that a polynomial time algorithm for \emph{approximately} computing the chromatic polynomial of a planar graph at non-real algebraic $q$ (satisfying some properties) leads to a polynomial time algorithm for \emph{exactly} computing it, which is known to be hard by a result of Vertigan.
Many of our results extend in fact to the more general partition function of the random cluster model, a well known reparametrization of the Tutte polynomial.
\end{abstract}


\begin{subject}
\end{subject}

\begin{document}
\section{Introduction}
The study of (approximately) computing the chromatic polynomial, or in fact the more general Tutte polynomial\footnote{Recall that the Tutte polynomial is a $2$-variable polynomial that has the chromatic polynomial among its many specializations.} was initiated by~\cite{complexityJonesandTutte} over thirty years ago. Among other things they proved that evaluating the  chromatic polynomial of a graph at any algebraic number $q$ exactly is \textsc{\#P}-hard except for $q=0,1,2$.
This was extended by~\cite{VertiganplanarTutte} who showed that the same is true when restricted to planar graphs.
The next step was taken by~\cite{GoldbergJerrum08} who proved that it is NP-hard to approximate the chromatic polynomial at real values $q>2$ (as part of a much larger result concerning inapproximability of the Tutte polynomial).
As far as we know the complexity of approximating the chromatic polynomial at real $q$ on \emph{planar graphs} is open. See~\cite{GoldbergJerrumplanar} for hardness result for evaluations of the Tutte polynomial on planar graphs `close' to the chromatic polynomial. 

Partly motivated by applications to quantum computing, ~\cite{GoldbergGuocomplexityIsing-Tutte} proved the first inapproximability results for certain non-real evaluations of the Tutte polynomial, showing \textsc{\#P}-hardness of approximating and not just NP-hardness. These results were recently extended in~\cite{GalanisGoldbergHerrera} to a much larger family of evaluations and planar graphs.

So far, as far as we know, no inapproximability results were known for the chromatic polynomial at non-real values of $q$. 
For several graph polynomials, such as the independence polynomial and the partition function of the Ising model, recent developments in the study of approximate counting has indicated that approximating evaluations of these polynomials is computationally hard in the vicinity of the zeros of these polynomials~\cite{PetersRegts,BGGSindsetcomplexplane,Buys,RoederLeeYangCayley,PetersRegts20,LeeYangzerosandcomplexity,galanis2022a}.
Motivated by this connection and Sokal's famous result~\cite{Sokaldense} saying that the zeros of of the chromatic polynomial are dense in the complex plane, one might be tempted to conjecture that approximating the chromatic polynomial at non-real numbers should be hard.
Our main result indeed confirms this.

\subsection{Main results}
Before we state our main result we first formally state the computational problems we are interested in and give the definition of the chromatic polynomial.

We denote the chromatic polynomial of a graph $G=(V,E)$ by $Z(G;q)$; it is defined as 
\[
Z(G;q) \coloneqq \sum_{A\subseteq E}(-1)^{|A|}q^{k(A)},
\]
where $k(A)$ denotes the number of components of the graph $(V,A)$.
For a positive integer $q$, $Z(G;q)$ equal the number of proper $q$-colorings of $G$.

We will consider two types of approximation problems, one for the norm of $Z(G;q)$ and one for its argument, for each algebraic number $q$ separately. For a nonzero complex number $\xi$ we will consider the argument $\arg(\xi)$ as an element of $\R/(2\pi \Z)$ measuring the angle with the positive real axis.  For $\overline{a}\in \R/(2\pi \Z)$ we denote $|\overline{a}|:=\min_{a'\in \overline{a}}|a'|.$

Let $\xi$ be a complex number and $\eta>0$. 
We call a number $r\in \mathbb{Q}$ an \emph{$\eta$-abs-approximation} of $\xi$ if $\xi\neq 0$ implies $e^{-\eta}\leq r/\left|\xi \right|\leq e^\eta$.
We call a number $r\in \mathbb{Q}$ an \emph{$\eta$-arg-approximation} of $\xi$ if $\xi\neq 0$ implies that $|r-\arg(\xi)| \leq \eta$.
Note that in both cases an approximation of $0$ could be anything.
In~\ref{subsec:algebraic numbers} below we will indicate how we will represent algebraic numbers.
Throughout graphs may have multiple edges between any pair of vertices and loops, unless stated otherwise.
Consider for an algebraic number $q$ the following computational problems. 
\\% for general graphs and planar graph separately.

\noindent
\begin{tabular}{rl}
    Name:& \textsc{$q$-Planar-Abs-Chromatic}\\
    Input:& A planar graph $G$.\\
    Output:& An $0.25$-abs-approximation of $Z(G;q).$
\end{tabular}\\[1em]
\noindent
\begin{tabular}{rl}
    Name:& \textsc{$q$-Planar-Arg-Chromatic}\\
    Input:& A planar graph $G$.\\
    Output:& An $0.25$-arg-approximation of $Z(G;q).$ 
\end{tabular}\\[1em]

\noindent We define the problems \textsc{$q$-Abs-Chromatic},\textsc{$q$-Arg-Chromatic} in the same way except that the input for both problems may now be any graph.
We note that these problems do not change in complexity when restricting to simple graphs, since the chromatic polynomial of a graph with a loop is constantly equal to $0$ and the chromatic polynomial of a graph with no loops is equal to the chromatic polynomial of its underlying simple graph.

Our main result is the following:
\begin{theorem}\label{thm:main_result planar}
    For each non-real algebraic number $q\in\C$ such that $|1-q|>1$ or $\Re(q)>3/2$, the problems \textsc{$q$-Planar-Abs-Chromatic} and  \textsc{$q$-Planar-Arg-Chromatic} are \textsc{\#P}-hard.
\end{theorem}
Note that by planar duality, this result also applies to the flow polynomial. 

As an immediate consequence of~\ref{thm:main_result planar}, we obtain hardness for approximately computing the chromatic polynomial on the entire complex plane except the real line for the family of all graphs:
\begin{corollary}\label{cor:main_result}
    For each non-real algebraic number $q\in\C$, the problems \textsc{$q$-Abs-Chromatic} and  \textsc{$q$-Arg-Chromatic} are \textsc{\#P}-hard.
\end{corollary}
\begin{proof}
This follows the same argument as Sokal's density result~\cite{Sokaldense}. We reduce the problems to their planar counterpart. 
Given a planar graph $G$. 
Clearly, we may assume that $|q-1|\leq 1$.
Create a new graph $\hat G$ by adding three new vertices pairwise connected by an edge and connect each of  these three vertices to all original vertices of $G$.
It is well known and easy to see that $Z(\hat G;q)=q(q-1)(q-2)Z(G;q-3)$.
Denote $q'=q-3$ and note that $|q'-1|=|(q-1)-3|>1$.
So a polynomial time algorithm that solves the problem \textsc{$q$-Abs-Chromatic} (resp. \textsc{$q$-Arg-Chromatic}) can be used to solve \textsc{$q'$-Planar-Abs-Chromatic} (resp. \textsc{$q'$-Planar-Arg-Chromatic}) in polynomial time.
Since the latter two problems are \textsc{\#P}-hard by~\ref{thm:main_result planar}, the same holds for the former two.
\end{proof}

While the main focus of this paper is on the chromatic polynomial, we also derive results for the more general partition function of the \emph{random cluster model} 
\[
Z(G;q,y) \coloneqq \sum_{A\subseteq E} (y-1)^{|A|} q^{k(A)},
\]
and the associated problems of approximating the value for fixed $q,y$ on input of a planar graph $G$. 
In~\ref{thm:general main result} we find a sufficient condition such that these problems are \textsc{\#P}-hard, and in~\ref{cor:##P-hard parameters} we record some explicit ranges for $q,y$ where the problems are \textsc{\#P}-hard, which includes the result of~\ref{thm:main_result planar}. ~\ref{fig:chromatic hardness region} shows a region of $q$-values for which we could verify the condition in~\ref{thm:general main result} with a computer.
\begin{figure}
    \centering
    \includegraphics[width=10cm]{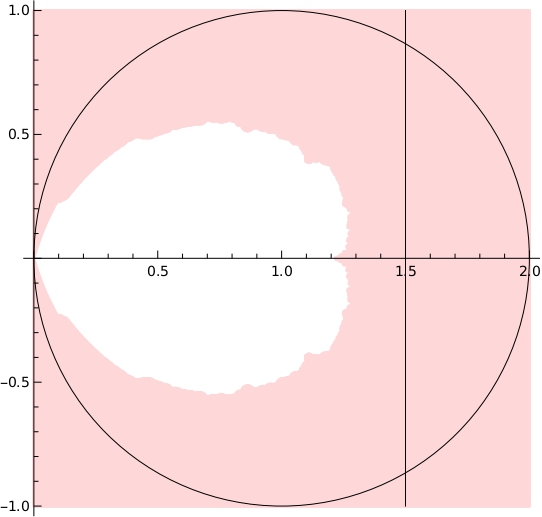}
    \caption{The red-shaded region represents values of $q$ for which the problems \textsc{$q$-Planar-Abs-Chromatic} and  \textsc{$q$-Planar-Arg-Chromatic} are \textsc{\#P}-hard. This is a pixel-picture with a $1001\times 1001$ resolution. The region depicted in the figure ranges from $-i$ to $2+i$.}
    \label{fig:chromatic hardness region}
\end{figure}

As is well known, $Z(G;q,y)$ is essentially equal to the Tutte polynomial $T(G;x,y)$ for $q=(x-1)(y-1)$.
See e.g.~\cite{Sokalsurvey,ellis2022handbook} for background. 
Thus \ref{thm:general main result} immediately gives us several inapproximability results for the Tutte polynomial on planar graphs. 
However, the cases $x=1$ or $y=1$, corresponding to $q=0$, are not covered by our approach. 
(The case $x=1$ corresponds to the all terminal reliability polynomial.)
We comment on how our approach could be used to show to hardness of approximation for these cases in~\ref{sec:conclusion}.
%In the special case of the chromatic polynomial, where $y=0$, 

\subsection{Proof outline}
We sketch a proof outline here (only for the chromatic polynomial), leaving most of the technical details and definitions to later sections.
Our proof goes along similar lines as other inapproximability results for non-real parameters obtained recently~\cite{GoldbergGuocomplexityIsing-Tutte,BGGSindsetcomplexplane,LeeYangzerosandcomplexity,GalanisGoldbergHerrera,galanis2022a}, but it also differs from these at certain steps.
For example, in the previous works just mentioned the problem of approximate evaluation is reduced from exact evaluation of the polynomial/partition function at different parameters (often real) causing extra work to be done.
One of our novel contributions is that we reduce exact evaluation at $q$ (which by~\cite{VertiganplanarTutte} is \textsc{\#P}-hard) to approximate evaluation at the same parameter, yielding a very clean reduction. 
Our approach for this is quite robust and could be applied to other partition functions/graph polynomials.
We say a bit more about this below.
First we describe our approach.

The goal is to show that for a given algebraic number $q$, assuming the existence of a polynomial time algorithm for \textsc{$q$-Planar-Abs-Chromatic} or \textsc{$q$-Planar-Arg-Chromatic}, we can design a polynomial time algorithm to compute the evaluation of the chromatic polynomial at $q$ exactly. This is essentially done in two steps.

The first step is to replace an edge $e$ of a given graph $G$ by another graph $H$ (with two marked vertices) also called a \emph{gadget}.
The chromatic polynomial of the resulting graph $G'$ is then, up to some easily computable factor (if the graph $H$ is series-parallel for example), given by
\[
Z(G;q)+y_HZ(G/e;q),
\]
where is $y_H$ the \emph{effective edge interaction} of $H$ (to be defined in the next section).
If one can determine the value $y^*$ such that $Z(G;q)+y^*Z(G/e;q)=0$, then this means that one can determine the ratio $r=\frac{Z(G;q)}{Z(G/e;q)}$, assuming $Z(G/e;q)\neq 0$.
A potential problem is when both $Z(G/e;q)$ and $Z(G\setminus e;q)$ are equal to zero. Below we indicate how to overcome this difficulty.

In case $q$ is real the value of $y^*$ can be approximated very accurately by means of a binary search procedure due to~\cite{GoldbergJerrumsign}. 
Then, using that $y^*$ is an algebraic number of polynomial size, this implies one can in fact determine $y^*$ \emph{exactly} in polynomial time. 
The binary search is done by applying the assumed polynomial time algorithm that approximately computes the absolute value or argument of $Z(G';q)$ and using the output of this algorithm for various values of $y$ to steer the binary search.
We extend and simplify this binary search strategy to a procedure we call `box shrinking', see~\ref{thm:box shrinking} below, so that it applies to non-real $q$. 
Moreover, we modify it in way so that it allows us to determine if $Z(G/e;q)$ is equal to $0$ or not, provided not both $Z(G/e;q)$ and $Z(G\setminus e,q)$ are zero.
Having this extra information allows us then to compute $Z(G;q)$ exactly, writing it as a telescoping product.
See~\ref{thm:allow telescope} below for the details.

The box shrinking procedure requires us to be able to generate graphs $H$ such that their effective edge interactions approximate any given $y_0\in \mathbb{Q}[i]$ with very high precision fast. 
This brings us to the second step. 
As in~\cite{GalanisGoldbergHerrera}, we use series-parallel graphs to achieve this.
The approach to do this is partly based on~\cite{BGGSindsetcomplexplane} and inspired by~\cite{de2021zeros,bencs2022location}. See~\ref{thm:gadget finding} below for the precise statement of what we obtain.

The novel parts of our approach, the box shrinking procedure, \ref{thm:box shrinking}, and~\ref{thm:allow telescope}, are quite robust and can for example be applied to the independence polynomial.
Doing so, one reduces the problem of exactly evaluating the independence polynomial at some non-real fugacity parameter, which is \textsc{\textsc{\#P}}-hard by~\cite{CKholant}, to approximately evaluating it.
This would shorten and simplify the reduction used in~\cite{BGGSindsetcomplexplane}, where a reduction from evaluating the polynomial at $1$ is used.

\paragraph{Organization}
In the next section we collect some definitions and results around series-parallel graphs and some notions regarding the representation of algebraic numbers.
Then in~\ref{sec:proof} we state our two main technical contributions and combine these to give a proof of~\ref{thm:main_result planar}.
~\ref{sec:constructing gadgets} and~\ref{sec:complex binary search} contain our proofs of these two contributions.
Finally in~\ref{sec:conclusion} we conclude with some questions and problems left open by our work.

\section{Preliminaries}
In this section we set up some notation and introduce some basic notions that we will use.
\subsection{Series-parallel graphs}\label{sec:SP}
As mentioned in the introduction we will be using series-parallel graphs as gadgets.
We introduce these here closely following~\cite{bencs2022location} and thereby~\cite{RoyleSokal} in their use of notation.
%`We refer to~\cite{} for further background.

Let $G_1$ and $G_2$ be two graphs with designated start- and endpoints $s_1,t_1$, and $s_2,t_2$ respectively, referred to as \emph{two-terminal graphs}.
The \emph{parallel composition} of $G_1$ and $G_2$ is the graph $G_1\parallel G_2$ with designated start- and endpoints $s,t$ obtained from the disjoint union of $G_1$ and $G_2$ by identifying $s_1$ and $s_2$ into a single vertex $s$ and by identifying $t_1$ and $t_2$ into a single vertex $t.$
The \emph{series composition} of $G_1$ and $G_2$ is the graph $G_1\bowtie G_2$ with designated start- and endpoints $s,t$ obtained from the disjoint union of $G_1$ and $G_2$ by identifying $t_1$ and $s_2$ into a single vertex and by renaming $s_1$ to $s$ and $t_2$ to $t$. Note that the order matters here.
A two-terminal graph $G$ is called \emph{series-parallel} if it can be obtained from a single edge using series and parallel compositions. 
Note that series-parallel graphs are automatically connected. 
%We wish to make one exception and call a two-terminal graph just consisting of the vertices $s$ and $t$ with no edges a series-parallel graph as well.

From now on we will implicitly assume the presence of the start- and endpoints when referring to a two-terminal graph $G$.
We denote by $\mathcal{G}_{\mathrm{SP}}$ the collection of all series-parallel graphs. % and by $\mathcal{G}^*_{\mathrm{SP}}$ the collection of all series-parallel graphs $G$ such that the vertices $s$ and $t$ are not connected by an edge.

\subsection{Effective edge interactions}
An important ingredient in our proof will be the notion of \emph{effective edge interaction}.
It requires a few preliminary definitions to define it. 
We extend these definitions from~\cite{bencs2022location} to the partition function of the random cluster model.

Recall that for a positive integer $q$, any $y\in\mathbb{C}$ and a graph $G=(V,E)$ we have
\[
Z(G;q,y)=\sum_{\phi:V\to \{1,\ldots,q\}} \prod_{uv\in E}(1+(y-1)\delta_{\phi(u),\phi(v)}),
\]
where $\delta_{i,j}$ denotes the Kronecker delta.
For a positive integer $q$ and a two-terminal graph $G$, we can thus write,
\begin{equation}\label{eq:same dif}
Z(G;q,y)=Z^{\mathrm{same}}(G;q,y)+Z^{\mathrm{dif}}(G;q,y),
\end{equation}
where $Z^{\mathrm{same}}(G;q,y)$ collects those contributions where $s,t$ receive the same color and where $Z^{\mathrm{dif}}(G;q,y)$ collects those contribution where $s,t$ receive different colors. 
Since $Z^{\mathrm{same}}(G;q,y)$ is equal to $Z(G';q,y)$ where $G'$ is obtained from $G$ by identifying the vertices $s$ and $t$, both these terms are polynomials in $q$ and $y$. 
Therefore \ref{eq:same dif} also holds for any $q\in \mathbb{C}$.

% We next collect some basic properties of $Z$, $Z^{\mathrm{same}}$ and $Z^{\mathrm{dif}}$ under series and parallel compositions in the lemma below. 
% They can for example also be found in~\cite{Sokaldense}.
% \begin{lemma}\label{lem:basic}
% Let $G_1$ and $G_2$ be two two-terminal graphs and let us denote by $K_2$ an edge. Then we have the following equalities:
% \begin{itemize}
%  \item   $Z^{\mathrm{dif}}(G;q) = Z(G \parallel K_2 ;q)$,
%     \item $Z^{\mathrm{same}}(G_1 \bowtie G_2 ;q) = Z(G_1 \parallel G_2 ;q)$,
%   \item  $Z(G_1 \bowtie G_2 ;q) = \tfrac{1}{q} \cdot Z(G_1 ;q) \cdot Z(G_2 ;q)$,
%   \item $Z^{\mathrm{same}}(G_1 \parallel G_2 ;q) = \tfrac{1}{q}\cdot Z^{\mathrm{same}}(G_1 ;q) \cdot Z^{\mathrm{same}}(G_2 ;q)$,
%   \item  $Z^{\mathrm{dif}}(G_1 \parallel G_2 ;q) = \tfrac{1}{q(q-1)}\cdot Z^{\mathrm{dif}}(G_1 ;q) \cdot Z^{\mathrm{dif}}(G_2 ;q)$.
% \end{itemize}
% \end{lemma}

%We note that $Z^{\mathrm{same}}(G;q)$ is equal to $Z(G',q,0)$, where $G'$ is the (multi)graph obtained from $G$ by identifying $s$ and $t$. 
%Thus for $G\in \mathcal{G}_{\mathrm{SP}}$, $Z^{\mathrm{same}}(G;q)$ is constantly equal zero if and only if $G\notin\mathcal{G}*_{\mathrm{SP}}$.

For fixed $y\in \mathbb{C}$ the \emph{effective edge interaction} is defined as
\begin{equation}\label{eq:ratio}
y_G(q,y):=(q-1)\frac{Z^{\mathrm{same}}(G;q,y)}{Z^{\mathrm{dif}}(G;q,y)},
\end{equation}
which we view as a rational function in $q$ and hence it takes values in the Riemann sphere $\mathbb{C}\cup \{\infty\}$.
(It might be slightly confusing that both the function and one of its inputs are called $y$. This is because $y_G$ will play a role similar as $y$. The name effective edge interaction stems from the fact that replacing all edges of a graph $H$ with a $2$-terminal graph $G$ with effective edge interaction $y_G$ yielding the graph $H(G)$ has the property that $Z(H;q,y_G)=Z(H(G);q,y)$ up to a simple factor depending only on $G$.)
We note that in case $G$ contains an edge between $s$ and $t$, the rational function $q\mapsto y_G(q,0)$ is constantly equal to $0$.
If $q,y$ are clear from the context we may occasionally just write $y_G$ for the effective edge interaction.

For any $q\neq 0$ define the following M\"obius transformation 
\[
f_q(z):=1+\frac{q}{z-1}\]
and note that $f_q$ is defined on the Riemann sphere and is an involution, i.e. $f_q(f_q(z))=z$ for all $z\in \mathbb{C}\cup \{\infty\}$, as follows by a direct calculation.
Following~\cite{bencs2022location}, for a two-terminal graph $G$ with effective edge interaction $y_G$, we call $f_q(y_G)$ a \emph{virtual interaction}.

The next lemmas are Lemma 1 and 3 from~\cite{bencs2022location} and capture the behavior of the effective edge interactions under series and parallel compositions. Even though they were only stated for the chromatic polynomial in~\cite{bencs2022location}, their proof automatically extends to the more general setting of the partition function of the random cluster model.

\begin{lemma}[\cite{bencs2022location}]\label{lem:basic Zsame and Zdif}
Let $G_1$ and $G_2$ be two two-terminal graphs and let $q,y\in \mathbb{C}$. Then we have the following identities:\\[5pt]
%\begin{align*}
% \bullet&& Z^{\mathrm{same}}(G_1 \parallel G_2 ;q,y) &= \tfrac{1}{q}\cdot Z^{\mathrm{same}}(G_1 ;q,y) \cdot Z^{\mathrm{same}}(G_2 ;q,y),\\
% \bullet&& Z^{\mathrm{dif}}(G_1 \parallel G_2 ;q,y) &= \tfrac{1}{q(q-1)}\cdot Z^{\mathrm{dif}}(G_1 ;q,y) \cdot Z^{\mathrm{dif}}(G_2 ;q,y),\\
% \bullet&& Z(G_1 \bowtie G_2 ;q,y) &= \tfrac{1}{q} \cdot Z(G_1 ;q,y) \cdot Z(G_2 ;q,y),\\
% \bullet&& Z^{\mathrm{same}}(G_1 \bowtie G_2 ;q,y)&= Z(G_1 \parallel G_2 ;q,y)\\
% &&&=\tfrac{1}{q}\cdot Z^{\mathrm{same}}(G_1 ;q,y) \cdot Z^{\mathrm{same}}(G_2 ;q,y)\\
%  &&&+\tfrac{1}{q(q-1)}\cdot Z^{\mathrm{dif}}(G_1 ;q,y) \cdot Z^{\mathrm{dif}}(G_2 ;q,y),\\
% \bullet&& Z^{\mathrm{dif}}(G_1 \bowtie G_2 ;q,y) &=\tfrac{1}{q}\cdot Z^{\mathrm{same}}(G_1;q,y)\cdot Z^{\mathrm{dif}}(G_2;q,y)\\
% &&&+\tfrac{1}{q}\cdot Z^{\mathrm{dif}}(G_1;q,y)\cdot Z^{\mathrm{same}}(G_2;q,y)\\
%    &&&+\tfrac{q-2}{q(q-1)}\cdot Z^{\mathrm{dif}}(G_1;q,y)\cdot Z^{\mathrm{dif}}(G_2;q,y).
%\end{align*}
$\begin{array}{crcl}
 \bullet& Z^{\mathrm{same}}(G_1 \parallel G_2 ;q,y) &=& \tfrac{1}{q}\cdot Z^{\mathrm{same}}(G_1 ;q,y) \cdot Z^{\mathrm{same}}(G_2 ;q,y),\\[5pt]
 \bullet& Z^{\mathrm{dif}}(G_1 \parallel G_2 ;q,y) &=& \tfrac{1}{q(q-1)}\cdot Z^{\mathrm{dif}}(G_1 ;q,y) \cdot Z^{\mathrm{dif}}(G_2 ;q,y),\\[5pt]
 \bullet& Z(G_1 \bowtie G_2 ;q,y) &=& \tfrac{1}{q} \cdot Z(G_1 ;q,y) \cdot Z(G_2 ;q,y),\\[5pt]
 \bullet& Z^{\mathrm{same}}(G_1 \bowtie G_2 ;q,y)&=& Z(G_1 \parallel G_2 ;q,y)\\[5pt]
 &&=&\tfrac{1}{q}\cdot Z^{\mathrm{same}}(G_1 ;q,y) \cdot Z^{\mathrm{same}}(G_2 ;q,y)\\[5pt]
  &&+&\tfrac{1}{q(q-1)}\cdot Z^{\mathrm{dif}}(G_1 ;q,y) \cdot Z^{\mathrm{dif}}(G_2 ;q,y),\\[5pt]
 \bullet& Z^{\mathrm{dif}}(G_1 \bowtie G_2 ;q,y) &=&\tfrac{1}{q}\cdot Z^{\mathrm{same}}(G_1;q,y)\cdot Z^{\mathrm{dif}}(G_2;q,y)\\[5pt]
 &&+&\tfrac{1}{q}\cdot Z^{\mathrm{dif}}(G_1;q,y)\cdot Z^{\mathrm{same}}(G_2;q,y)\\[5pt]
    &&+&\tfrac{q-2}{q(q-1)}\cdot Z^{\mathrm{dif}}(G_1;q,y)\cdot Z^{\mathrm{dif}}(G_2;q,y).
\end{array}$
\end{lemma}

\begin{lemma}[\cite{bencs2022location}]\label{lem:formulas effective}
Let $G_1,G_2$ be two two-terminal graphs and let $y\in \mathbb{C}$. Then the following identities hold as rational functions:
\begin{align*}
y_{G_1\parallel G_2}&=y_{G_1}\cdot y_{G_2},
\\
f_q(y_{G_1\bowtie G_2})&=f_q(y_{G_1}) \cdot f_q(y_{G_2}).
 \end{align*}
Moreover, for any fixed $q_0,y_0\in \mathbb{C}$, if $\{y_{G_1}(q_0,y_0),y_{G_2}(q_0,y_0)\}\neq \{0,\infty\}$, then
\[
y_{G_1\parallel G_2}(q_0,y_0)=y_{G_1}(q_0,y_0)\cdot y_{G_2}(q_0,y_0),
\]
and if $ \{y_{G_1}(q_0,y_0),y_{G_2}(q_0,y_0)\} \neq \{1,1-q_0\}$ and $q_0\neq 0$, then
\[
f_{q_0}(y_{G_1\bowtie G_2}(q_0,y_0))=f_{q_0}(y_{G_1}(q_0,y_0))\cdot f_{q_0}(y_{G_2}(q_0,y_0)).
\]
\end{lemma}

\subsection{Representing algebraic numbers}\label{subsec:algebraic numbers}
We discuss here how we deal with the representation of algebraic numbers.

An algebraic number $a$ is by definition a complex number that is the zero of some polynomial with integer coefficients. The minimal polynomial of $a$ is the unique polynomial $p\in \mathbb{Z}[x]$ of smallest degree such that $p(a)=0$, whose coefficients have no common prime factors, and whose leading coefficient is positive when $a\neq 0$.
Following \cite{galanis2022a,GalanisGoldbergHerrera} we represent an algebraic number $a$ by its minimal polynomial together with an open rectangle in the complex plane (defined by two rational intervals parallel to the real and imaginary axis respectively) such that $a$ is the only zero of the polynomial in that rectangle. 
There is some ambiguity here, as many rectangles will do the job. However one can check whether two representations represent the same number, by checking if the polynomials are equal and checking if the two rectangles intersect and by counting the number of zeros in the intersection (using for example the algorithm of \cite{Wilfzero}).
When referring to an algebraic number we thus implicitly assume we have its minimal polynomial and a rectangle as above.
The \emph{size} of an algebraic number will thus be the number of bits needed to represent the minimal polynomial and the rectangle.  

It is described in~\cite{computingin} how to execute all basic operations, i.e., addition, subtraction, multiplication and inversion in terms of this representation.
It can be seen that these operations can be executed in polynomial time in the representation sizes by combining a result from~\cite{Mahlerrootdistance}, lower bounding the distance between roots of a polynomial, a result on using resultants to find polynomials with a prescribed zero~\cite{computingLoos}, the famous factoring algorithm of~\cite{LLLfactoring}, a result of~\cite{Mignotte} bounding the coefficients of factors of polynomials, and an algorithm of~\cite{Wilfzero} finding the number of zeros of a polynomial in a rectangle. 
It follows that we can compare and compute absolute values of algebraic numbers and test whether two algebraic numbers are equal in polynomial time in their representation size.

\section{Proof of the main results}\label{sec:proof}
In this section we state our main technical contributions, which we will prove in the subsequent sections. These results allow us to prove our main results at the end of this section.
Let us first define for algebraic numbers $q,y\in \mathbb{C}$ the more general computational problems \textsc{$(q,y)$-Planar-Abs-RC} and \textsc{$(q,y)$-Planar-Arg-RC}, which on input of a planar graph $G$ ask for an $0.25$-abs-approximation resp.\@ an $0.25$-arg-approximation to $Z(G;q,y)$. 
The problems \textsc{$q$-Planar-Abs-Chromatic} and \textsc{$q$-Planar-Arg-Chromatic} are the particular cases \textsc{$(q,0)$-Planar-Abs-RC} and \textsc{$(q,0)$-Planar-Arg-RC}.

\subsection{Telescoping}
For a graph $G$ and an edge $e$ of $G$ we denote by $G\setminus e$ the graph obtained from $e$ by removing the edge $e$ and by $G/e$ the graph obtained from $G$ by contracting the edge $e$
Recall that we allow multiple edges between two vertices and loops. 
We note here that the family of planar graphs is closed under deletion and contraction of edges.
We recall the well known deletion contraction recurrence for a graph $G$ and an edge $e$ of $G$ (contrary to the Tutte polynomial, this recurrence also holds when $e$ is a bridge or a loop):
\begin{equation}
Z(G;q,y)=Z(G\setminus e;q,y)+(y-1)Z(G/e;q,y). \label{eq:del-con}
\end{equation}
If one can determine for any graph $G$ and an edge $e\in E(G)$ one of the ratios $\frac{Z(G;q,y)}{Z(G/e;q,y)}$, or $\frac{Z(G;q,y)}{Z(G\setminus e;q,y)}$ exactly, then one can compute $Z(G;q,y)$ exactly by writing it as a telescoping product. 
A potential catch is that both $Z(G;q,y)$ and $Z(G/e;q,y)$ (and hence $Z(G\setminus e;q,y)$ by \ref{eq:del-con}) could equal zero.
Under suitable assumptions, our next result is able to deal with this issue.
In what follows we denote for a graph $H$ by $\size(H)$ the sum of the number of vertices of $H$ and the number of edges of $H$. 

\begin{theorem}\label{thm:allow telescope}
Let $q,y$ be fixed algebraic numbers, with $q\neq 0$.
Suppose that we have access to an algorithm that on input of a planar graph $G$ and an edge $e\in E(G)$ outputs an algebraic number $r$ and a number $b\in \{0,1\}$ in polynomial time in $\size(G)$ such that %for any edge $e$ of $G$,
\begin{itemize}
    \item[(1)] If $Z(G/e;q,y)\neq 0$, then $b=1$ and $r=\frac{Z(G;q,y)}{Z(G/e;q,y)}$;
    \item[(2)] if $Z(G/e;q,y)= 0$ and $Z(G\setminus e;q,y)\neq 0$, then $b=0$ and $r=1$;
    \item[(3)] if both $Z(G/e;q,y)$ and $Z(G\setminus e;q,y)$ are zero, then the algorithm may output any algebraic number $r$ and bit $b$.
    \end{itemize}
Then there is an algorithm to compute $Z(G;q,y)$ in polynomial time in $\size(G)$.    
\end{theorem}

\begin{proof}%[Proof of Theorem~\ref{thm:allow telescope}]
We construct a sequence of planar graphs $G_0,\ldots,G_m$, where $G_m$ is a graph with no edges, as follows.
We let $G_0=G$.
Now for $i\geq 0$, we apply the assumed algorithm to the graph $G_i$ and an edge $e_i$ of $G_i$.
Assume the algorithm outputs the pair $(r_i,b_i)$.
If $b_i=1$ we set $G_{i+1}=G_i/e_i$ and if $b_i=0$ we set $G_{i+1}=G_i\setminus e_i$.
Let $n$ denote the number of vertices of $G_m$.
The output of our algorithm will be the number $q^{n}\prod_{i=0}^{m-1} r_i$.
Since $m$ is at most the number of edges of $G$ and since $\size(G_i)\leq \size(G_0)$ for each $i$, this clearly takes polynomial time in $\size(G)$ to compute.

What remains to show is that 
\begin{equation}\label{eq:Z is product}
    Z(G;q,y)=q^{n}\prod_{i=0}^{m-1} r_i.
\end{equation}
To prove~\ref{eq:Z is product}, let us first assume that $Z(G_0;q,y)=0$. 
Then we do not know whether or not we can trust the output of the algorithm, as possibly both $Z(G/e_0;q,y)$ and $Z(G\setminus e_0;q,y)$ could be zero.
However, there is a smallest index $i$ such that $Z(G_i;q,y)\neq 0$ (since $Z(G_m;q,y)=q^n\neq 0$).
In this case we know that $Z(G_{i-1};q,y)=0$ and hence $b_{i-1}=1, G_i=G_{i-1}/e_i$ and thus $r_{i-1}=0$ as well.
Therefore~\ref{eq:Z is product} holds in this case.

Let us next assume that $Z(G_0;q,y)\neq 0$. In that case one of the two values $Z(G/e_0;q,y)$, $Z(G\setminus e_0;q,y)$ must be non-zero by~\ref{eq:del-con}.
This means we are in case (1) or (2) and hence if $b=1$, $G_1=G_0/e_0$, and it holds that $Z(G_1;q,y)\neq 0$ while if $b=0$, $G_1=G_0\setminus e_0$, and it holds that $Z(G_1;q,y)\neq 0$.
By induction it follows that for all $i$, $Z(G_i,q,y)\neq 0$ and thus $r_i\neq 0$ for all $i$.
%Therefore, $Z(G;q)=0$ if and only if $r_i=0$ for some index $i$.
Next observe that
\[
r_i=\frac{Z(G_i;q,y)}{Z(G_{i+1};q,y)}.
\]
Indeed, if $b_i=1$, then this follows by construction and if $b_i=0$ we have $Z(G_{i+1};q,y)=Z(G_i\setminus e_i;q,y)=Z(G_i;q,y)$ by~\ref{eq:del-con} since $Z(G_i/e_i;q,y)=0$.
Therefore, 
\[
\frac{Z(G;q,y)}{q^n}=\frac{Z(G_0;q,y)}{Z(G_1;q,y)}\cdot \frac{Z(G_1;q,y)}{Z(G_2;q,y)}\cdots  \frac{Z(G_{m-1};q,y)}{Z(G_m;q,y)}=\prod_{i=0}^{m-1}r_i.
\]
This finishes the proof.
\end{proof}

This result and its proof are fairly general and do not really rely on specific properties of the partition function of the random cluster model, but could equally well be applied to other polynomials and partition functions that satisfy some recurrence relation such as the independence polynomial for example.

\subsection{Implementing gadgets}
In the previous subsection we indicated that if one can test whether $Z(G/e;q,y)$ is zero or not and compute the ratios $\frac{Z(G;q,y)}{Z(G/e;q,y)}$ exactly this gives rise to an algorithm to exactly compute $Z(G;q,y)$.

The idea from~\cite{GoldbergJerrumsign} for real valued parameters is to use approximations to $\hat{y}Z(G/e;q,y)+Z(G;q,y)$ to steer a binary search procedure to obtain a very precise approximation to the value $y^*$ which satisfies $y^*Z(G/e;q,y)+Z(G;q,y)=0$, or in other words the ratio $-\frac{Z(G;q,y)}{Z(G/e;q,y)}$.
The following result describes the outcome of our box shrinking procedure, which extends this binary search procedure to complex parameters and moreover simplifies it.
(It may be helpful to think of $A=Z(G/e;q,y)$ and $B=Z(G;q,y)$ in the statement of the theorem below.)

For $r>0$ and $m\in \mathbb{C}$ we denote $B(m,r)=\{z\in\C \mid |z-m|<r\}$ and $B_\infty(m,r)=\{z\in \mathbb{C}\mid |\Re(z-m)|<r, |\Im(z-m)|<r\}$. Note that we can view $B_\infty$ as an open ball of radius $r$ centered at $m$ in the $\ell^\infty$-metric on $\mathbb{R}^2=\mathbb{C}.$
\begin{theorem}\label{thm:box shrinking}
Let $A,B$ be complex numbers and let $C>0$ be a rational number such that 
$|A|$ and $|B|$ are both at most $C$, and both are either $0$ or at least $1/C$.
Assume one of the following: 
\begin{itemize}
\item there exists a $\poly(\size(y_0,\eps))$-time algorithm to compute on input of $y_0\in\Q[i]$ and a rational number $\epsilon>0$ an $0.25$-abs-approximation of $A\hat{y}+B$ for some algebraic number $\hat{y} \in B(y_0,\eps)$, or,
\item there exists a $\poly(\size(y_0,\eps))$-time algorithm to compute on input of $y_0\in\Q[i]$ and a rational number $\epsilon>0$ an $0.25$-arg-approximation of $A\hat{y}+B$ for some algebraic number $\hat{y}\in B(y_0,\epsilon)$.
\end{itemize}
Then there exists an algorithm that on input of a rational $\delta>0$ and $C>0$ as above that outputs ``$A=0$'' when $A=0$ and $B\neq 0$, and that outputs ``$A\neq 0$'' and a number $\bar{y}\in \mathbb{Q}[i]$ such that $-B/A\in B_\infty(\bar{y},\delta/2)$ when $A\neq 0$. 
When $A=B=0$ it is allowed to output anything. 
The running time is $\poly(\size(C,\delta))$.
\end{theorem}

We will prove~\ref{thm:box shrinking} in~\ref{sec:complex binary search}.
To utilize it we must be able to generate for any given value $y_0$ a number $\hat{y}$ that approximates $y_0$ with arbitrary precision in a way that we can approximate the norm or argument of $\hat{y}Z(G/e;q,y)+Z(G;q,y)$ efficiently.
The idea, going back to~\cite{GoldbergGuocomplexityIsing-Tutte}, is to use certain series-parallel gadgets for this task.

For fixed values of $q,y$, we call a series-parallel graph $H$ a \emph{series-parallel gadget for $(q,y)$} if $Z^\mathrm{dif}(H;q,y)\neq0$. %(Note that any series-parallel graph is a gadget for all but finitely many $q$.)
Let $G$ be a graph with a designated edge $e=\{u,v\}$, and let $H$ be a series-parallel gadget for $(q,y)$.
Then we construct the graph $G'$ obtained from $G$ and $H$ by removing the edge $e$ from $G$ and identifying the start vertex of $H$ with $u$ and the terminal vertex with $v$. Note that by flipping $u$ and $v$ this may result in a different graph. We call any such graph an \emph{implementation of $H$ in $G$ on $e$}.
See~\ref{fig:sp} for an illustration.
%Let $G'$ be a graph obtained from $G$ by replacing the edge $e$ by the graph $H$.

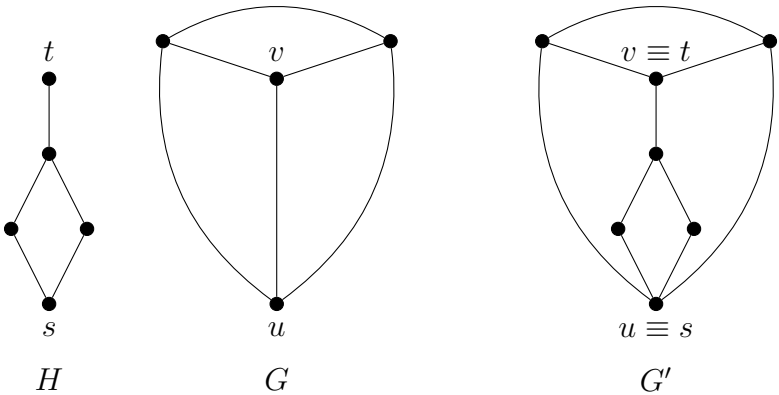
\begin{figure}[h!]
    \centering
     \begin{tikzpicture}
           \node[vertex,label=below:$s$] (u) at (0,0) {};
           \node (captionH) at ($(u)+(0,-1)$) {$H$};
           \node[vertex,label=above:$t$] (v) at (0,3) {};
           \node[vertex] (h1) at (-0.5,1) {};
           \node[vertex] (h2) at (0.5,1) {};
           \node[vertex] (h3) at (0,2) {};
           
           \draw (u) -- (h1);
           \draw (u) -- (h2);
           \draw (h1) -- (h3);
           \draw (h2) -- (h3);
           \draw (v) -- (h3);
           
           \node (trans0) at (3,0) {};

           \node[vertex,label=below:$u$] (g1) at ($(u)+(trans0)$) {};
           \node[vertex,label=above:$v$] (g2) at ($(v)+(trans0)$) {};
           \node[vertex] (g3) at ($(-1.5,3.5)+(trans0)$) {};
           \node[vertex] (g4) at ($(1.5,3.5)+(trans0)$) {};
           
           \node (captionG) at ($(captionH)+(trans0)$) {$G$};
           
           \draw (g1) -- (g2);
           \draw (g1) to[bend left] (g3);
           \draw (g1) to[bend right] (g4);
           \draw (g2) -- (g3);
           \draw (g2) -- (g4);
           \draw (g3) to[bend left] (g4);
           
           \node (trans) at (5,0) {};
           \node (trans2) at ($(trans)+(trans0)$) {};
           
           \node[vertex,label=below:$u \equiv s$] (gg1) at ($(g1)+(trans)$) {};
           \node[vertex,label=above:$v\equiv t$] (gg2) at ($(g2)+(trans)$) {};
           \node[vertex] (gg3) at ($(g3)+(trans)$) {};
           \node[vertex] (gg4) at ($(g4)+(trans)$) {};
           \node[vertex] (gh1) at ($(h1)+(trans2)$) {};
           \node[vertex] (gh2) at ($(h2)+(trans2)$) {};
           \node[vertex] (gh3) at ($(h3)+(trans2)$) {};
          \node (captionGH) at ($(captionH)+(trans0)+(trans)$) {$G'$};

           \draw (gg1) -- (gh1);
           \draw (gg1) -- (gh2);
           \draw (gh1) -- (gh3);
           \draw (gh2) -- (gh3);
           \draw (gg2) -- (gh3);
           \draw (gg1) to[bend left] (gg3);
           \draw (gg1) to[bend right] (gg4);
           \draw (gg2) -- (gg3);
           \draw (gg2) -- (gg4);
           \draw (gg3) to[bend left] (gg4);

     \end{tikzpicture}
    \caption{Implementation of $H$ in $G$ on the edge $\{u,v\}$.}
    \label{fig:sp}
\end{figure}

The next result says that with access to an algorithm that efficiently solves \textsc{$(q,y)$-Planar-Arg-RC} or \textsc{$(q,y)$-Planar-Abs-RC}, we can efficiently approximate $\hat{y}Z(G/e;q,y)+Z(G;q,y)$ for any value $\hat{y}+y$ that is the effective edge interaction of a series-parallel gadget.
\begin{lemma}\label{lem:implementing gadget alg}
%Given a series-parallel gadget $H$ with $y_H(q)=y$, and a 
Let $(q,y)\in\mathbb{C}^2$ be algebraic numbers, such that $q\not\in\{0,1\}$.
Assume there exists a $\poly(\size(G))$-time algorithm to compute an $0.25$-abs-approximation (resp.\@ $0.25$-arg-approximation) of $Z(G;q,y)$ for planar graphs. 
Then there exists a $\poly(\size(G,H))$-time algorithm to find an $0.25$-abs-approximation (resp.\@ $0.25$-arg-approximation) of $(y_H(q,y)-y)Z(G/e;q,y)+Z(G;q,y)$ on input of any planar graph $G$, edge $e$ of $G$, series-parallel gadget $H$ and the number $Z^\mathrm{dif}(H;q,y)$.
\end{lemma}

\begin{proof}
We will see $G$ as a two-terminal graph, with the endpoints of $e$ being the terminals.
In this case by definition, we have that
\begin{align*}
 Z^\mathrm{dif}(G\setminus e;q,y)&=Z^\mathrm{dif}(G;q,y)\\
 &=Z(G;q,y)-Z^\mathrm{same}(G;q,y) \\
 &= Z(G;q,y)-y Z(G/e;q,y).
\end{align*}
Let $G'$ denote an implementation of $H$ in $G$ on $e$ and observe that $G'$ is planar. 
We compute using~\ref{lem:basic Zsame and Zdif},
\begin{align*}
    Z(G';q,y) &= \tfrac{1}{q}Z^\mathrm{same}(G\setminus e;q,y) Z^\mathrm{same}(H;q,y) 
    \\& \tfrac{1}{q(q-1)}Z^\mathrm{dif}(G\setminus e;q,y)Z^\mathrm{dif}(H;q,y)\\
    &=\tfrac{Z^\mathrm{dif}(H;q,y)}{q(q-1)}\left((q-1)\tfrac{Z^\mathrm{same}(H;q,y)}{Z^\mathrm{dif}(H;q,y)} Z(G/e;q,y)\right.
    \\&+Z(G;q,y)- yZ(G/e;q,y)\big)\\
    &=\tfrac{Z^\mathrm{dif}(H;q,y)}{q(q-1)}\left((y_H(q,y)-y) Z(G/e;q,y) +Z(G;q,y)  \right).
\end{align*}
As $H$ is a gadget we know that $Z^{\mathrm{dif}}(H;q,y)\neq 0$, and we have its exact value as input of the algorithm.
%and since it is series-parallel we can compute  $Z^{\mathrm{dif}}(H;q,y)$ exactly in time polynomial in $|H|$ (for example by constructing its series-parallel decomposition~\cite{RecognizeSP} in combination with Lemma~3 from~\cite{bencs2022location}).

We then use the assumed algorithm to compute an $0.25$-abs-approximation of $Z(G';q,y)$. 
Multiplying the output by $\left|\frac{q(q-1)}{Z^{\mathrm{dif}}(H;q,y)}\right|$ produces an $0.25$-abs-approximation of \[(y_H(q,y)-y) Z(G/e;q,y) +Z(G;q,y).\] Because $\size(G')=\size(G)+\size(H)-3$ this will be done in $\poly(\size(G,H))$-time, as desired.

Similarly, if we use the assumed algorithm to compute an $0.25$-arg-approximation of $Z(G';q,y)$ and add $\arg \left(\frac{q(q-1)}{Z^{\mathrm{dif}}(H;q,y)}\right)$ to the output, we obtain an $0.25$-arg-approximation of \[(y_H(q,y)-y) Z(G/e;q,y) +Z(G;q,y)\] in $\poly(\size(G,H))$-time.
\end{proof}

In order to obtain a desired algorithm for~\ref{thm:box shrinking} using~\ref{lem:implementing gadget alg} we need to be able to quickly approximate any desired value $y_0$ with effective edge interactions of gadgets. This final ingredient is~\ref{thm:gadget finding}, which will be proved in~\ref{sec:constructing gadgets}.

To do this precisely, we limit ourselves to a slightly smaller family of series-parallel graphs.
Define $\mathcal{H}_{q,y}^*$ to be the family of all series-parallel graphs $H$ that satisfy $y_H(q,y)\not\in \{1,\infty\}$ and $Z^\mathrm{dif}(H;q,y)\neq 0$.

\begin{theorem}\label{thm:gadget finding}
Let $(q,y)\in \mathbb{C}^2\setminus \mathbb{R}^2$ be algebraic numbers such that $q\not\in\{0,1\}$.
Assume there exists $G\in \mathcal{H}^*_{q,y}$ for which $y_G(q,y)$ or $f_q(y_G(q,y))$ is finite and has absolute value strictly bigger than $1$.
There exists an algorithm that for every $y_0\in\mathbb{Q}[i]$ and rational $\epsilon>0$ outputs a series-parallel graph $H\in \mathcal{H}_{q,y}^*$ such that $y_H(q,y) \in B(y_0,\epsilon)$, and outputs the number $Z^\mathrm{dif}(H;q,y)$.
Both the running time of the algorithm and the size of the graph are $\poly(\size(\eps,y_0))$. 
\end{theorem}

%\begin{theorem}\label{thm:gadget escape special cases}
%Let $q\in \mathbb{C}\setminus \mathbb{R}$ be an algebraic number such that $|q-1|>1$ or $\Re(q)>3/2$. 
%Then there exists an algorithm that on input of an algebraic number $y\in\mathbb{C}$ and a rational $\epsilon>0$ computes a series-parallel gadget $H$ with $|y_H(q)-y| < \epsilon$. Both the size of $H$ and the running time of the algorithm are $\poly(\size(y),\size(\eps))$.
%\end{theorem}

\subsection{Proof of~\ref{thm:main_result planar}}\label{subsec:proof of main thm}
In this section we combine all ingredients mentioned above to provide a proof of~\ref{thm:main_result planar}

To utilize these ingredients, we need to introduce some concepts regarding algebraic numbers. 
We closely follow~\cite{GalanisGoldbergHerrera} in doing so.

We will introduce these concepts for polynomials with integer coefficients, and then automatically obtain definitions for algebraic numbers by applying the definition to their minimal polynomial. In particular, we can talk about the degree $d(\alpha)$ of an algebraic number $\alpha$.
We further define several variants of the height of a polynomial $p\in\Z[x]$: the usual height $H(p)$ is the largest absolute value among the coefficients, and the length $L(p)$ is the sum of the absolute value of the coefficients (similarly we define these notions for multi-variable polynomials).
If we let $d$ be the degree of $p$, let $a_d$ be the leading coefficient of $p$ and let $\alpha_i$ be the roots of $p$, we define the Mahler measure $M(p) \coloneqq |a_d|\prod_{i=1}^d\max(1,|\alpha_i|)$ and finally the absolute logarithmic height $h(p) \coloneqq \frac{1}{d}\log (M(p))$.
Here and in what follows, $\log$ denotes the natural logarithm.

From the definitions we can deduce for any algebraic number $\alpha$ that $d(1/\alpha)=d(\alpha)$, $h(1/\alpha)=h(\alpha)$, and for every rational function $f\in\Q(x)$ that $d(f(\alpha))\leq d(\alpha)$.
%We also see that for any two algebraic numbers $\alpha,\beta$ it holds that $h(\alpha\beta)\leq h(\alpha)+h(\beta)$.
The next lemma records some more intricate results about these heights, and gives some bounds for algebraic numbers of bounded height. Note in particular that item (b) yields the special case $h(\alpha_1\alpha_2) \leq h(\alpha_1)+h(\alpha_2)$ for any algebraic numbers $\alpha_1,\alpha_2$.

\begin{lemma}\label{lem:results algebraic numbers}
\begin{enumerate}[(a)]
\item For any non-zero algebraic number $\alpha$ we have $-d(\alpha)h(\alpha)\leq\log (|\alpha|) \leq d(\alpha)h(\alpha)$.
\item Let $p\in\Z[x_1,\ldots,x_t]$ be a polynomial, $\alpha_1,\ldots,\alpha_t$ algebraic numbers, and define the algebraic number $\beta=p(\alpha_1,\ldots,\alpha_t)$. Write $d_i(p)$ for the degree of $p$ as polynomial in $x_i$, then $h(\beta) \leq \log (L(p)) + \sum_{i=1}^t d_i(p)h(\alpha_i)$.
\item For any polynomial $p\in\Z[x]$ we have
\[
\tfrac{1}{d(p)}\log (H(p)) -\log(2) \leq h(p). %\leq \tfrac{1}{d}\log H(p) + \tfrac{1}{2d}\log(d+1).
\]
\end{enumerate}
\end{lemma}
\begin{proof}
All parts can be found in \cite{waldschmidt2013diophantine}:
parts (a) and (b) are respectively in Section~3.5.1 and Lemma~3.7; part (c) is in Lemma~3.11 for the special case that $p$ is irreducible, but the proof works for any polynomial $p$.
\end{proof}

\begin{corollary}\label{cor:bounding algebraic numbers}
Let $q,y$ be algebraic numbers and let $G$ be a graph with $n$ vertices and $m$ edges.
Write $d,h_q,h_y$ for respectively $d(q)\cdot d(y),h(q),h(y)$.
\begin{enumerate}[(a)]
\item If $Z(G;q,y)\neq 0$, then \[\big|\log (|Z(G;q,y)|) \big|\leq \left(nh_q+mh_y + 2m\log (2) \right)d.\]
\item Let $e$ be an edge of $G$ and assume that $Z(G/e;q,y)\neq 0$.
%and let $y^*=Z(G;q)/Z(G/e;q)$.% and assume that not both $Z(G;q)/Z(G/e;q)$. 
Then
\[
\log (H\left(\frac{Z(G;q,y)}{Z(G/e;q,y)}\right))\leq \left(2nh_q+2mh_y+4m\log (2) \right)d.
\]
\end{enumerate}
\end{corollary}

\begin{proof}
Let $G=(V,E)$ be our graph, then let us recall the definition $Z(G;q,y)=\sum_{A\subseteq E}q^{k(A)}(y-1)^{|A|}$.
From this we see that $Z(G;q,y)$ is a polynomial of degree $n$ in $q$ and degree $m$ in $y-1$. 
In the variables $q$ and $y-1$, we also see that the sum of the absolute value of the coefficients is $2^m$. 
Note that part (b) of the previous lemma gives $h(y-1)\leq h(y)+\log (2)$, and applying it again to $Z$ gives $h(Z(G;q,y))\leq m\log(2)+n h_q + m(h_y+\log(2))$.
Further $d(Z(G;q,y))\leq d$.
Then by part (a) when $Z(G;q,y)\neq 0$, we see that $\big|\log(|Z(G;q,y)|)\big| \leq (nh_q + mh_y + 2m\log(2))d$, proving part (a).

For part (b) we first look at the absolute logarithmic height of the ratio. We have
\begin{align*}
h\left(\frac{Z(G;q,y)}{Z(G/e;q,y)}\right)&\leq h(Z(G;q,y)) + h(Z(G/e;q,y))\\
&\leq (2n-1)h_q + (2m-1)h_y + (4m-2)\log(2).
\end{align*}
Again the ratio has degree at most $d$, so part (c) of the previous lemma gives 
\begin{align*}
\log \left(H\left(\frac{Z(G;q,y)}{Z(G/e;q,y)}\right)\right) &\leq \left(h\left(\frac{Z(G;q,y)}{Z(G/e;q,y)}\right) + \log(2)\right)d\\ &\leq \left(2nh_q+2mh_y+4m\log (2) \right)d. \qed
\end{align*}
\end{proof}

The following result follows from a slight modification of a result of~\cite{kannan1984polynomial}.

\begin{proposition}\label{prop:determine minimal polynomial}
Let $d,H\in \mathbb{N}_{\geq 2}$ and $\bar{\alpha}\in\mathbb{Q}[i]$.
There exists an algorithm that on input $d,H$ and $\bar{\alpha}$ outputs a polynomial $p$ such that if there exists an algebraic number $\alpha$ of degree at most $d$, height at most $H$ such that $\log(|\alpha-\bar{\alpha}|)\leq -(d^2+5d+2d\log (H))$, then $p$ is the minimal polynomial of $\alpha$. The running time is bounded by $\poly(d,\log (H),\size(\bar{\alpha}))$.
\end{proposition}
\begin{proof}
We first describe the algorithm and then prove its correctness.

We first assume $|\bar{\alpha}|\leq 1$. 
We will apply Algorithm~1.16 from \cite{kannan1984polynomial}. 
If the algorithm terminates within $d$ steps and outputs a polynomial $p$, we output $p$. If after $d$ steps the algorithm does not return anything we output $0$. 
If $|\bar{\alpha}|\geq 1$ we replace $\bar{\alpha}$ by $1/\bar{\alpha}$ and we refer to the proof~\cite[Theorem 1.19]{kannan1984polynomial} how to modify the output of Algorithm~1.16 in this case.

To prove correctness, we may assume $|\bar{\alpha}|\leq 1$ and we may further assume there exist an algebraic number $\alpha$ of degree at most $d$ and height at most $H$ such that 
\begin{equation}\label{eq:alpha close}
|\alpha-\bar{\alpha}|\leq  e^{-(d^2+5d)}\cdot H^{-2d}.
\end{equation}
Otherwise there is nothing to prove.

We next note that as $d\ge 2$ we have
\[
    d^2+5d \ge \log(48d(d+1)^{(3d+4)/2}2^{d^2/2}),
\]
therefore there exists a non-negative integer $s\geq 2$ such that
\begin{equation}
     \tfrac{1}{12d} e^{d^2+5d} \cdot H^{2d}\ge 2^s \ge 2(d+1)^{(3d+4)/2}2^{d^2/2}\cdot H^{2d}. \label{eq:bound dH}
\end{equation}

If $|\alpha|\leq 1$ the conclusion follows immediately from the correctness of~\cite[Algorithm~1.16]{kannan1984polynomial}.
Unfortunately we do not have this information, but we will argue that even in the case $|\alpha|>1$ the algorithm is still correct. We will show that by assuming that a bound of $1+1/(2d)$ on $|\alpha|$ is sufficient. 
By~\ref{eq:alpha close} and \ref{eq:bound dH} this bound is clearly satisfied and it implies that for any $k=1,\ldots,d$, $|\alpha|^k\leq 2$.
With this bound on $|\alpha|$~\cite[Proposition 1.6]{kannan1984polynomial} remains true if we multiply the right-hand side by $1/2$.
By our slightly stronger lower bound on $s$ we see that also the conclusion of~\cite[Lemma 1.9]{kannan1984polynomial} is also valid for this bound on $|\alpha|$.
Finally as in~\cite[Explanation 1.17]{kannan1984polynomial} the conditions of~\cite[Theorem 1.15]{kannan1984polynomial} are still met for this bound on $|\alpha|$. 
Therefore the output of~\cite[Algorithm 1.16]{kannan1984polynomial} is indeed the minimal polynomial of $\alpha$, as desired

The running time bound follows from (the proof of) Theorem 1.19 from~\cite{kannan1984polynomial} where we note that~\cite{kannan1984polynomial} in fact does not require the input of the number $\bar{\alpha}$, but only a certain number of bits of it and therefore the size of $\bar{\alpha}$ does not appear in the running time of the statement of the theorem.
To avoid dealing with how to feed $\bar{\alpha}$ to the algorithm we just allow the algorithm to `read' it completely. \end{proof}

%\subsubsection{Proof of Theorem~\ref{thm:main_result planar}}
 %Now we are ready to prove Theorem~\ref{thm:main_result}. The main idea of the proof is that we compute the ratio $Z(G;q)/Z(G/e;q)$ exactly, and then use a telescoping product to compute $Z(G;q)$. There are however some subtleties that arise when some quantity is zero.

% \begin{proof}[Proof of Theorem~\ref{thm:main_result}]
% We first recall the deletion-contraction relation
% \[
% Z(G;q)=Z(G\setminus e;q) - Z(G/e;q).
% \]
% In particular this tells us, if $Z(G/e;q)=0$, then $Z(G;q)=Z(G\setminus e;q)$.
Now we are ready to prove the following result, from which we will derive~\ref{thm:main_result planar} after giving the proof.
For the theorem let us recall that $\mathcal{H}_{q,y}^*$ is the family of all series-parallel graphs $H$ that satisfy $Z^\mathrm{dif}(H;q,y)\neq 0$ and $y_H(q,y)\not\in \{1,\infty\}$.
\begin{theorem}\label{thm:general main result}
Let $(q,y)\in \mathbb{C}^2\setminus \mathbb{R}^2$ be algebraic numbers, such that $q\not\in\{0,1,2\}$.
Assume there exists $G\in \mathcal{H}^*_{q,y}$ for which $y_G(q,y)$ or $f_q(y_G(q,y))$ is finite and has absolute value strictly bigger than $1$.
Then the problems \textsc{$(q,y)$-Planar-Abs-RC} and \textsc{$(q,y)$-Planar-Arg-RC} are \textsc{\#P}-hard.
\end{theorem}

\begin{proof}
We will assume there exists a polynomial-time algorithm for either the problem \textsc{$(q,y)$-Planar-Abs-RC} or \textsc{$(q,y)$-Planar-Arg-RC}, and then we will find a $\poly(\size(G))$-time algorithm to compute $Z(G;q,y)$ exactly for planar graphs $G$.
Since the latter problem is \textsc{\#P}-hard by \cite{VertiganplanarTutte}, apart from several exceptional values, this implies that \textsc{$(q,y)$-Planar-Abs-RC} and \textsc{$(q,y)$-Planar-Arg-RC} are \textsc{\#P}-hard as well.
The exceptional values are $y=1$, $q=0,1,2$, $(q,y)=(3,e^{\pm\frac{2\pi i}{3}})$ and $(q,y)=(4,-1)$. 
We explicitly excluded most of them in the statement, while for $y=1$ the family $\mathcal{H}_{q,y}^*$ is empty, and for $(q,y)=(3,e^{\pm\frac{2\pi i}{3}})$ we can show that $y_G(q,y)$ is contained in $\{e^{\frac{2\pi i}{3}},e^{-\frac{2\pi i}{3}}\}$ for any $G\in\mathcal{H}_{q,y}^*$ (see the proof of~\ref{cor:##P-hard parameters} below).

The proof essentially consists of linking together~\ref{thm:gadget finding}, ~\ref{lem:implementing gadget alg}, ~\ref{thm:box shrinking}, ~\ref{cor:bounding algebraic numbers}, ~\ref{prop:determine minimal polynomial} and ~\ref{thm:allow telescope}, roughly in that same order.

%It is our plan to combine the assumed algorithm with two previously mentioned algorithms to devise an algorithm as in the assumption of Theorem~\ref{thm:allow telescope}. The result then follows from that theorem.

To keep track of the running times of all the separate algorithms we describe and analyze the resulting algorithm in one go and prove its correctness afterwards.

The assumption in the theorem means that we can apply \ref{thm:gadget finding} to find a gadget $F$ with effective edge interaction $y_F$ close to $y_0+y$ for any $y_0$ (note that $\size(y_0+y)=O(\size(y_0))$). We use $F$ and $Z^\mathrm{dif}(F;q,y)$ in \ref{lem:implementing gadget alg} to approximate $(y_F-y)Z(G/e;q,y)+Z(G;q,y)$. Combined we obtain an algorithm that on input of $y_0\in\Q[i]$, rational $\epsilon >0$, planar graph $G$ and an edge $e$, will compute an $0.25$-abs-approximation (resp.\@ $0.25$-arg-approximation) to $\hat{y} Z(G/e;q,y)+Z(G;q,y)$ for some algebraic number $\hat{y} \in B(y_0,\epsilon)$, in $\poly(\size(G,y_0,\epsilon))$-time.

As in \ref{cor:bounding algebraic numbers}, we write $h_q,h_y$ for the absolute logarithmic height of $q$ and $y$, and $d=d(q)d(y)$.
We now want to apply \ref{thm:box shrinking} with $A=Z(G/e;q,y)$ and $B=Z(G;q,y)$.
To do so we must find an common upper bound for them.
Note that by \ref{cor:bounding algebraic numbers}, we have that $|\log(|A|)|$ and $|\log(|B|)|$ are both bounded by $(nh_q+mh_y+2m\log (2))d$, where $n$ resp.\@ $m$ denote the number of vertices resp.\@ edges of $G$.
Now taking $C$ such that $\log (C)=(nh_q+mh_y+2m\log (2))d$ we satisfy the assumption of \ref{thm:box shrinking}. 
We thus have that $C=O(\size(G))$ (since $q,y$ and hence $h_q,h_y,d$ are considered to be constant).

Take $H$ such that $\log(H)=2(nh_q+mh_y+2m\log (2))d$.
We now apply the algorithm from \ref{thm:box shrinking} with $\log(\delta^{-1})=(d^2+5d)+2d\log(H)$ and as output we get either the statement that ``$A=0$'', or  ``$A\neq 0$'' and a number $\bar{y}$ such that $Z(G;q,y)/Z(G/e;q,y)\in B_\infty(-\bar{y},\delta/2)$.
The running time is bounded by $\poly(\size(C,\delta))=\poly(\size(G))$ (using that $h_q,h_y,d$ are constants).

We now turn this algorithm into an algorithm as required in \ref{thm:allow telescope}.
In case we get ``$A=0$'', we output the pair $(0,1)$.
In case we get ``$A\neq0$'', we run the algorithm of \ref{prop:determine minimal polynomial} on input of $d,H$ and $-\bar{y}$ and let $p$ be the output of this algorithm.
We output the pair $(1,(p,B_\infty(-\bar{y},\delta/2)))$.
The running time is $\poly(\size(G))$, using that $h_q,h_y,d$ are constants and the size of $\bar{y}$ is $\poly(\size(G))$.
This implies that the overall running time is $\poly(\size(G))$, as desired.

We next turn to proving correctness of our algorithm. 
What remains is to show that the algorithm as required in \ref{thm:allow telescope} is correct.
If both $Z(G/e;q,y)$ and $Z(G\setminus e;q,y)$ are $0$, there is nothing to prove.
So let us first assume that $Z(G/e;q,y)=0$ and $Z(G\setminus e;q,y)\neq 0$.
Then by correctness of \ref{thm:box shrinking} we know that our algorithm gives the desired output, namely the pair $(0,1)$. 
Similarly, if $Z(G/e;q,y)\neq 0$, then by \ref{cor:bounding algebraic numbers}(b) we know that $y^*=-Z(G;q,y)/Z(G/e;q,y)$ is an algebraic number of degree at most $d$ and height at most $H$ and by correctness of \ref{thm:box shrinking} it is contained in $B_\infty(\bar{y},\delta/2))$.
Therefore, by our choice of $\delta$, \ref{prop:determine minimal polynomial} implies that $p$ is indeed the minimal polynomial of $-y^*$ in this case.
By a result of \cite{Mahlerrootdistance}, the absolute value of the logarithm of the distance between any two distinct zeros of $p$ is upper bounded by 
\begin{align*}
-\log(\sqrt{3}) + \tfrac{d+2}{2}\log(d)&+(d-1)\log((d+1)H) \\&\leq \tfrac{3}{2}d\log(d+1)+d\log (H)
\leq \log(\delta^{-1}).
\end{align*}
Therefore the rational rectangle $B_\infty(-\bar{y},\delta/2)$ together with $p$ is a representation of the algebraic number $-y^*$, and so the output of our algorithm is as desired.

This finishes the proof.
\end{proof}

In the following corollary, we indicate several regions of $q,y$-parameters for which \ref{thm:general main result} applies. Part (a) includes the (complex-valued) ferromagnetic Potts model. \ref{thm:main_result planar} is included in part (b). Note that part (b) also contains Theorem~1.5 in \cite{GalanisGoldbergHerrera}, our methods merely give a different proof.
Note however that part (b) is not tight, as \ref{fig:chromatic hardness region} shows a larger region of $q$-values where both problems are \textsc{\#P}-hard.

\begin{corollary}\label{cor:##P-hard parameters}
Let $(q,y)\in\C^2\setminus \R^2$ be algebraic numbers, such that $q\not\in\{0,1,2\}$. The problems \textsc{$(q,y)$-Planar-Abs-RC} and \textsc{$(q,y)$-Planar-Arg-RC} are \textsc{\#P}-hard in both of the following cases:
\begin{enumerate}[(a)]
    \item $|y|>1$;
    \item $|1-q|>1$ or $\Re(q)>3/2$, except when $y=1$ or $(q,y)=(3,e^{\pm\frac{2\pi i}{3}})$.
\end{enumerate}
\end{corollary}

\begin{proof}
For part (a), we simply have to note that $K_2 \in \mathcal{H}_{q,y}^*$ and $y_{K_2}(q,y)=y$, so the graph $K_2$ satisfies the requirement of \ref{thm:general main result}.

For part (b), the assumption yields that either $f_q(0)=1-q$ or $f_q(f_q(0)^2)=\frac{q-1}{q-2}$ will have absolute value strictly bigger than $1$. This means that if we find a $y_G$ or a $f_q(y_G)$ close to these values, we can apply \ref{thm:general main result}.

In what follows we implicitly use \ref{lem:formulas effective} and \ref{prop:non-zero dif} to see that certain numbers are the virtual or effective edge interaction of some series-parallel graph in $\mathcal{H}_{q,y}^*$.
We assume $y\neq 1$, which ensures that the family $\mathcal{H}_{q,y}^*$ is non-empty. If there now exists a graph $G$ in this family with $|y_G|>1$, we are done. However, if $|y_G|<1$, its powers converge to $0$. This means that for $N$ large enough, one of $f_q(y_G^N)$ and $f_q(f_q(y_G^N)^2)$ will have absolute value strictly bigger than $1$, and both are the (virtual) interaction of a series-parallel graph in $\mathcal{H}_{q,y}^*$. So we are left with the case that $|y_G|=1$ for all $G\in\mathcal{H}_{q,y}^*$.
Applying the same argument to $f_q(y_G)$, we may even assume that $|f_q(y_G)|=1$ for all $G\in\mathcal{H}_{q,y}^*$. 

Geometrically, the equations $|z|=1$ and $|f_q(z)|=1$ give a circle and a line in the complex plane, so there can be at most two common solutions $z$. 
This means that the sets $\{y_G(q,y) \mid G\in\mathcal{H}_{q,y}^*\}$ and $\{f_q(y_G(q,y)) \mid G\in\mathcal{H}_{q,y}^*\}$ both contain at most two elements. 
From \ref{lem:formulas effective} and \ref{prop:non-zero dif} below it follows that both sets are closed under taking products, as long as the product is not $1$. 
Therefore the only options for both sets are $\{-1\}$ and $\{e^{\frac{2\pi i}{3}} , e^{-\frac{2\pi i}{3}} \}$, corresponding to $(q,y) \in \{(4,-1), (3,e^{\frac{2\pi i}{3}}), (3,e^{-\frac{2\pi i}{3}}) \}$. All these options are excluded, so we conclude that this special case cannot occur.
\end{proof}

\section{Constructing series-parallel gadgets}\label{sec:constructing gadgets}
In this section we will prove \ref{thm:gadget finding}.
%We will in fact prove a stronger statement from which Theorem~\ref{thm:gadget escape} follows as a direct consequence.

\subsection{A family of potential gadgets}\label{subsec:family H of potential gadgets}
We start by exhibiting a family of series-parallel graphs that can serve as gadgets.

Let $(q,y)\in \mathbb{C}^2$ and recall that we denote by $\mathcal{H}_{q,y}^*$ the family of all series-parallel graphs $H$ that satisfy $Z^\mathrm{dif}(H;q,y)\neq 0$ and $y_H(q,y)\not\in \{1,\infty\}$.
%largest family $\mathcal{H}$ of series-parallel graphs that satisfy the following two assumptions:
%\begin{itemize}
%    \item If $H\in\mathcal{H}$ with $H=H_1 \parallel H_2$ or $H=H_1 \bowtie H_2$, then both $H_1$ and $H_2$ are in $\mathcal{H}$.
%    \item For every $H\in\mathcal{H}$, we have $y_H(q,y)\notin\{1,\infty\}$.
%\end{itemize} 
%Note that $\mathcal{H}_{q,y}^*$ is indeed unique; if $\mathcal{H}_1$ and $\mathcal{H}_2$ are two maximal such families, then $\mathcal{H}_1\cup \mathcal{H}_2$ also satisfies the assumptions. 
Note that if $y=1$, the family $\mathcal{H}_{q,y}^*$ is always empty; while for $y\neq 1$ and $q\not\in\{0,1\}$, the edge $K_2$ is contained in $\mathcal{H}_{q,y}^*$.
The next result says that this family is more or less closed under taking parallel and series composition.
%states that elements of $\mathcal{H}_{q,y}^*$ can in fact serve as gadgets.

\begin{proposition}\label{prop:non-zero dif}
Let $(q,y)\in\C^2$ such that $q\not\in \{0,1\}$, and let $H_1,H_2\in\mathcal{H}_{q,y}^*$.
\begin{itemize}
    \item If $y_{H_1 \parallel H_2}(q,y) \not \in \{1,\infty\}$, then $H_1\parallel H_2 \in \mathcal{H}_{q,y}^*$.
    \item If $y_{H_1 \bowtie H_2}(q,y) \not \in \{1,\infty\}$, then $H_1\bowtie H_2 \in \mathcal{H}_{q,y}^*$.
\end{itemize}
\end{proposition}

\begin{proof}
Write $H=H_1 \parallel H_2$. \ref{lem:basic Zsame and Zdif} yields
\begin{align*}
    Z^\mathrm{dif}(H;q,y) &= \tfrac{1}{q(q-1)} \cdot Z^\mathrm{dif}(H_1;q,y) \cdot Z^\mathrm{dif}(H_2;q,y).
\end{align*}
If $Z^\mathrm{dif}(H;q,y)=0$, then $Z^\mathrm{dif}(H_i;q,y)=0$ for some $i\in \{1,2\}$, which is a contradiction. So we conclude that $Z^\mathrm{dif}(H;q,y)\neq 0$, and hence that $H\in\mathcal{H}_{q,y}^*$.
% induction hypothesis
% implies that $Z^\mathrm{same}(H_1;q)\neq 0$, so $y_{H_1}(q)=\infty$ which is a contradiction.\\

The proof for $H=H_1\bowtie H_2$ is similar, we first recall from \ref{lem:basic Zsame and Zdif}:
\begin{align*}
    Z(H;q,y) &= \tfrac{1}{q} \cdot Z(H_1;q,y)\cdot Z(H_2;q,y),\\
    Z^\mathrm{same}(H;q,y) &=\tfrac{1}{q} \cdot Z^\mathrm{same}(H_1;q,y) \cdot Z^\mathrm{same}(H_2;q,y) \\
    &\quad+ \tfrac{1}{q(q-1)} \cdot Z^\mathrm{dif}(H_1;q,y) \cdot Z^\mathrm{dif}(H_2;q,y).
\end{align*}
Suppose that $Z^\mathrm{dif}(H;q,y)=0$. This implies that $Z^\mathrm{same}(H;q,y)=0$, since $y_H\neq\infty$.
Therefore $Z(H;q,y)=0$, and it follows without loss of generality that $Z(H_2;q,y)=0$. So $Z^\mathrm{same}(H_2;q,y)=-Z^\mathrm{dif}(H_2;q,y)$, and they are non-zero by assumption. Plugging this into the second equation, together with $Z^\mathrm{same}(H;q,y)=0$, yields $(q-1)Z^\mathrm{same}(H_1;q,y)=Z^\mathrm{dif}(H_1;q,y)$. We also assumed that $Z^\mathrm{dif}(H_1;q,y)\neq 0$, leading to $y_{H_1}(q,y)=1$, which is a contradiction. Therefore we can again conclude that $H\in\mathcal{H}_{q,y}^*$.
\end{proof}

\subsection{Values of \texorpdfstring{$q,y$}{q,y} with dense set of interactions}
In this subsection we determine for which values of $q,y$ the effective edge interactions of members of $\mathcal{H}_{q,y}^*$ are dense in the complex plane. 
The next subsection deals with making this algorithmic yielding a proof of \ref{thm:gadget finding}.

The following lemma will turn out to be useful to get density results. We say that a set $S\subseteq \C$ is $\epsilon$-dense if for every $z\in\C$ there exists a $z'\in S$ such that $|z-z'|<\epsilon$.
\begin{lemma}\label{lem:dense linear combinations}
Let $\eps>0$ and let $a,b,c\in B(0,\epsilon)$ such that the convex cone spanned by $a,b,c$ is $\mathbb{C}$. Then the set $a\mathbb{N}+b\mathbb{N}+c\mathbb{N}$ is $\epsilon$-dense in $\mathbb{C}$.
\end{lemma}
\begin{proof}
For any $z\in\mathbb{C}$ we can write $z=ra+sb+tc$ with $r,s,t\geq0$. We can even assume that one of $r,s,t$ is zero, say $t=0$. Now we round $r,s$ to the nearest integers $R,S$. Then $Ra+Sb\in a\mathbb{N}+b\mathbb{N}$ and
\[
|Ra+Sb - z| \leq |R-r||a|+|S-s||b| < \tfrac{1}{2}\epsilon + \tfrac{1}{2}\epsilon = \epsilon.\qed
\]
\end{proof}

Recall the definition $f_q(z)=1+\frac{q}{z-1}$. Also when $G$ is a two-terminal graph, we call $f_q(y_G(q,y))$ its \emph{virtual interaction}.
\begin{proposition}\label{prop:dense theta}
Let $(q,y)\in\mathbb{C}^2\setminus \mathbb{R}^2$, such that $q\not\in\{0,1\}$.
Assume there exists a series-parallel graph $G\in \mathcal{H}^*_{q,y}$ such that $y_G(q,y)$ or $f_q(y_G(q,y))$ is finite and has absolute value strictly bigger than $1$, then the set $\{y_H(q,y)\mid H\in \mathcal{H}^*_{q,y}\}$ is dense in $\mathbb{C}$.
\end{proposition}

\begin{proof}%[Proof of Proposition~\ref{prop:dense theta}]
We start by showing that we may assume that $f_q(y_G)\notin \mathbb{R}$ and  $|f_q(y_G)|>1$ for some $G\in\mathcal{H}_{q,y}^*$. 
In what follows we implicitly use \ref{lem:formulas effective} and \ref{prop:non-zero dif} to conclude that certain numbers are the virtual or effective edge interaction of some series-parallel graph in $\mathcal{H}_{q,y}^*$.

First assume that $|y_G|>1$ and $y_G$ is real. 
Let $\xi_n=y_G^n$ for $n\in\N$.
If $y\not\in \mathbb{R}$, then  $y\xi_n$ is non-real and it is converges to infinity as $n\to \infty$. 
If $y\in \mathbb{R}$, then $q\not\in\R$ and $f_q(f_q(y)^2)=1+\frac{(y-1)^2}{2y+q-2}$ is also non-real (note that $y \neq 1$ because $\mathcal{H}_{q,y}^*$ is non-empty). 
This means that $f_q(f_q(y)^2)\xi_n$ is non-real and converges to infinity as $n\to \infty$. 
In either case we find a non-real term in the sequence with absolute value bigger than $1$, which is the effective edge interaction of a graph in $\mathcal{H}_{q,y}^*$. The precise graphs is a parallel composition of sufficiently many copies of $G$, together with either one edge or one path of length 2.

For the second case, assume that $|y_G|>1$ and $y_G$ is not real. Then powers of $y_G$ converge to $\infty$ and moreover the arguments (modulo $2\pi$) of these powers take on at least $3$ values. 
If the number of these values is unbounded, it is clear that for some $n$, $f_q(y_G^n)$ is non real and $|f_q(y_G^n)|>1$.
Otherwise, if these values take on precisely $k\geq 3$ values, the values $y_G^n$ converge towards $\infty$ on the $k$ rays from $0$ through $e^{2\pi j/k}$, $j=0,\ldots,k-1$.
The images of these rays under $f_q$ are circular arcs from $f_q(\infty)=1$ to $f_q(0)=1-q$ making pairwise angles of $2\pi/k$ at $1$. Since $k\geq 3$, at least one of these arcs contains a open segment starting at $1$ that does not intersect the closed unit disk. 
Any element from this segment has absolute value bigger than $1$, thus there exists a series-parallel graph $G'\in \mathcal{H}^*_{q,y}$ such that $|f_q(y_{G'})|>1$, that is the parallel composition of sufficiently many copies of $G$.

Let $\tau$ be a non-real element of $\{f_q(y),f_q(y^2)\}=\{1+\tfrac{q}{y-1},1+\tfrac{q}{(y-1)(y+1)}\}$. Then $\tau \cdot \xi^n$ is not real, converges to infinity, and all are virtual interactions of graphs in $\mathcal{H}_{q,y}^*$ corresponding to series composition of copies of $G$ with either an edge or a digon (two parallel edges). Therefore we will find one of absolute value more than $1$, as desired.

From now on we thus assume that we have $G\in \mathcal{H}_{q,y}^*$ such that $f_q(y_G)$ is not real and $|f_q(y_G)|>1$.
Consider next the M\"obius transformation 
\begin{equation}\label{eq:mobius g}
g(z):=f_q(f_q(z)f_q(y_G))=\frac{q-1+zy_G}{q-2+z+y_G}.
\end{equation}
Note that $z=1$ is a fixed point of $g$ and that $g'(1)=1/f_q(y_G)$.
So by assumption $|g'(1)|<1$.

Now consider the sequence defined by $y_1=y_G$ and $y_{i+1}=g(y_i)$, which converges to $1$. 
The only reason that this could not be true is if $y_G$ would be equal to the other (repelling) fixed point of $g$.
The other fixed point of $g$ is given by $1-q$. 
Now $y_G\neq 1-q$ since $f_q(1-q)=0$, while $|f_q(y_G)|>1$.

We next claim that taking finite products of terms in this sequence, produces a dense subset of $\mathbb{C}$.
%The $2$-terminal graphs corresponding to this sequence are series connections of $G$ by Lemma~\ref{lem:formulas effective} and are contained in $\mathcal{H}_{q,y}^*$.
Indeed, note that 
\[\frac{\log(y_{i+1})}{\log(y_i)} =\frac{\log(g(y_i))-\log(g(1))}{\log(y_i)-\log(1)}\to g'(1)=1/f_q(y_G)
\]
as $i\to\infty$. 
Since the number $1/f_q(y_G)$ is non-real, at some point the difference between the arguments of the $\log(y_i)$ lie in a small interval around $\arg(1/f_q(y_G))$, which is nonzero $\mod \pi$. 
In particular, this means that for any $k$ the argument of the sequence $(\log(y_i))_{i\geq k}$ cannot be contained in a half plane. Then for every $\epsilon>0$, we can find three terms in the sequence that satisfy the requirements of \ref{lem:dense linear combinations}. 
% and from this conclude that the effective edge interactions of generalized theta-graphs are dense in $\mathbb{C}$.
This lemma then implies that products of the values $y_i$ are dense in $\mathbb{C}$. 
These products correspond exactly to parallel compositions of series connections of $G$ by \ref{lem:formulas effective}. 
By omitting any effective edge interactions equal to $1$ possibly obtained by taking products of the $y_i$ we still obtain density and since the corresponding $2$-terminal graphs are contained in $\mathcal{H}_{q,y}^*$, this finishes the proof.
\end{proof}

\subsection{Implementing a dense set fast}
In the previous subsection we saw for which values of $q,y$ we can obtain density.
We will show here how to exploit this algorithmically, by proving \ref{thm:gadget finding}.

First we record the following lemma from~\cite{galanis2022a}, which in fact is essentially Lemma 2.8 from~\cite{BGGSindsetcomplexplane} adapted to algebraic numbers.
A M\"obius transformation $\Phi$ is called \emph{contracting} on a set $U\subset \mathbb{C}$ if $|\Phi'(z)|<1$ for all $z\in U$.
%For $z_0\in \mathbb{C}$ and $r>0$ we denote by $B(z_0,r)$ the set of complex numbers $z$ such that $|z-z_0|<r.$
\begin{lemma}\label{lem:maps covering U}
Let $m\in\mathbb{Q}[i]$ and $r>0$ rational. 
Further suppose that we have M\"obius transformations with algebraic coefficients $\Phi_i:\hat{\mathbb{C}} \to \hat{\mathbb{C}}$ for $i\in[t]$, satisfying the following with $U=B(m,r)$:
\begin{enumerate}
    \item for each $i\in[t]$, $\Phi_i$ is contracting on $U$;
    \item $U \subseteq \bigcup_{i=1}^t \Phi_i(U)$.
\end{enumerate}
There is an algorithm which, on input of algebraic numbers $x_0,x_1\in U$ (respectively the target and starting point) and rational $\epsilon >0$, outputs in $\poly(\size(x_0,x_1,\epsilon))$-time an algebraic number $\hat{x}\in B(x_0,\epsilon)$ and a sequence $i_1,\ldots,i_k\in[t]$ such that
\begin{align*}
\hat{x}=\Phi_{i_k}&\circ \dots\circ \Phi_{i_1}(x_1), & k&=O(\log(\eps^{-1})),
\\ 
\text{and } \Phi_{i_j}&\circ\dots\circ \Phi_{i_1}  (x_1)\in U & \text{for all } j&=1,\ldots,k.
\end{align*}
\end{lemma}
Technically the last requirement in the lemma is not stated as such in~\cite{galanis2022a} nor~\cite{BGGSindsetcomplexplane}, but it follows immediately from the proof.

%A \emph{generalized theta-graph} is a series-parallel graph consisting of a number of internally vertex disjoint paths between the start and terminal vertex.
To describe the algorithm for \ref{thm:gadget finding}, we first precompute some data which only depends on $q,y$ (and not on $y_0$ or $\epsilon$), so this counts as constant time. %Using this data, we describe the algorithm.
In the precomputation step, we find the following:
\begin{itemize}
    \item a rational number $r>0$ so that $U:=B(1,r)$ is an open around $1$;
    \item maps $\Phi_i$ as in the above lemma.%, which are of the form $z\mapsto f_q(f_q(z)f_q(0))y_i$, where $y_i$ is the effective edge interaction of a generalized theta-graph.
    %\item an integer $N$ such that $f_q(U^N)^2$ is all of $\hat{\mathbb{C}}$.
\end{itemize}
%Note that these maps are indeed M\"obius transformations with algebraic coefficients.
%Now we can easily find the maps as in Lemma~\ref{lem:maps covering U}. 
By the proof of \ref{prop:dense theta}, we may assume that $f_q(y_G(q,y))$ is non-real and has absolute value more than $1$ for some $G\in\mathcal{H}_{q,y}^*$.

Recall the definition of the M\"obius transformation $g$ in~\ref{eq:mobius g} and recall from the proof of \ref{prop:dense theta} that $g(1)=1$ and $|g'(1)|<1$. 
We can therefore choose an $\alpha$ such that $|g'(1)| < \alpha < 1$. 
Then take $U$ to be a ball $B(1,r)$ contained in the open set
\[
\{u\in\mathbb{C} \mid \alpha < |u| < 1/\alpha, |g'(u)|<\alpha\}.
\]
This ensures that for any $u\in U$, the map $z\mapsto g(z)u$ is a contraction on $U$.

Restricting to effective edge interactions $y_H(q,y)$ contained in $U$ for $H\in \mathcal{H}_{q,y}^*$, we see that by \ref{prop:dense theta} the open sets $g(U)y_H$ cover the compact set $\bar{U}$. 
By the proof of \ref{prop:dense theta} we can select a finite set $\{y_i\mid i\in I\}$ in finite time such that $\cup_{i\in I }g(U)y_i$ already covers $\bar{U}$. 
Now the $\Phi_i$ are defined to be the maps $z\mapsto g(z)y_i$.
Let us fix for each $i\in I$ a series-parallel graph $H_i\in \mathcal{H}_{q,y}^*$ whose effective edge interaction is equal to $y_i$. Using \ref{lem:basic Zsame and Zdif}, we also compute $Z^\mathrm{same}(H_i;q,y)$ and $Z^\mathrm{dif}(H_i;q,y)$.
% Observe that by construction, each of the paths in a generalized theta-graph $H_i$ corresponding to the effective edge interaction $y_i$ has effective edge interaction not equal to $1$ (since the map $(g(y)$ is a bijection and $g(1)=1$). 
% We take such $H_i$ with a minimal number of edges.
% We may therefore assume that $H_i$ is contained in $\mathcal{H}^*.$ (Otherwise the effective edge interaction of some subcollection of $s-t$ paths is equal to $1$ and we could just remove those.)

%The last part in the pre-computation is to find an $N$ such that $f_q(U^N)^2$ is $\hat{\mathbb{C}}$. For $N$ large enough, the power $U^N$ misses just small opens around 0 and $\infty$, so $f_q(U^N)$ misses just small opens around 1 and $1-q$. Since $q$ is non-real, we can take $N$ large enough such that these two opens are on one side of a line through 0. The other side of the line is all reached by $f_q(U^N)$, so $f_q(U^N)^2$ is indeed all of $\hat{\mathbb{C}}$.
%We now consider the case $|q-1|<1$ for $q\in \mathbb{Q}[i]$.

We now give a proof of \ref{thm:gadget finding}.
\begin{proof}[Proof of \ref{thm:gadget finding}]
We first consider the case where $y_0\in U$. 
We take as starting point $y_1=1$ in $U$ and run the algorithm of \ref{lem:maps covering U}. This yields in $\poly(\size(y_0,\eps))$ time a sequence $i_1,\ldots,i_k$ with $\hat{y}=\Phi_{i_k}(\cdots \Phi_{i_1}(y_1)\cdots) \in B(y_0,\eps)$. (Recall that $k=O(\log(\eps^{-1}))$.) 
We may assume that $\Phi_{i_j}(\cdots \Phi_{i_1}(y_1)\cdots)\neq 1$ for all $j=1,\ldots,k$. 
Otherwise we replace the sequence by $i_{j+1},\ldots,i_k$ with $j$ the largest index for which $\Phi_{i_j}(\cdots \Phi_{i_1}(y_1)\cdots)=1$.
% By the proof of \ref{lem:maps covering U} from~\cite{BGGSindsetcomplexplane} we also have $\Phi_{i_j}(\cdots \Phi_{i_0}(y_0)\in U$ for all $j=1,\ldots,k$ and therefore we have $\Phi_{i_j}(\cdots \Phi_{i_0}(y_0)\cdots)\neq\infty$ for all $j$. 

From the sequence $i_1,\ldots,i_k$ we can determine a sequence of series-parallel graphs $G_1,\ldots,G_k$. The sequence starts with $G_1=H_{i_1}$, which has effective edge interaction $\Phi_{i_1}(1)=y_i$.
Recall that every map $\Phi_i$ is of the form $z\mapsto y_i g(z)$, which by \ref{lem:formulas effective} corresponds to a series composition with $G$, and a parallel composition with a graph in \{$H_i\mid i\in I\}$. So we let $G_j = (G_{j-1} \bowtie G) \parallel H_{i_j}$, then $H\coloneqq G_k$ has effective edge interaction $\hat{y}$.
Since $G$ and $H_i$ only depend on $q,y$ they have constant size, and therefore the size of $H$ is $O(\log (\eps^{-1}))$.

Using \ref{lem:basic Zsame and Zdif} we inductively compute $Z^\mathrm{same}$ and $Z^\mathrm{dif}$ of $G_j$ in $\poly(k)$-time, so we can output $Z^\mathrm{dif}(H;q,y)$ along with $H$.

%$H$ of size $O(k)=O(\log(\eps^{-1}))$ with effective edge interaction $\hat{y}$ in time linear in $k$, as the result of these series- and parallel-compositions.

% We have to check that $\hat{y}$ is indeed the effective interaction of a series-parallel gadget $H$ with $O(\log(1/\epsilon))$ edges. 

% Hence every map $\Phi_i$ adds a constant number of edges to the gadget, and the final size is $O(k)=O(\log(1/\epsilon))$. 
By construction $\hat{y}\in B(y_0,\eps)$, so to prove correctness of the algorithm in this case it suffices to show that $H$ is indeed a gadget, that is, $Z^{\mathrm{dif}}(H;q,y)\neq 0$.
To do so we will inductively show that $G_j\in\mathcal{H}_{q,y}^*$ using \ref{prop:non-zero dif} and thereby in particular that $Z^{\mathrm{dif}}(G_j;q,y)\neq 0$. 
%Let us denote the graphs corresponding to the sequence $i_1,\ldots,i_k$, by $G_1,\ldots,G_k=H$.
Note that $G_{1}=H_{i_1}\in \mathcal{H}_{q,y}^*$ by assumption. % as it corresponds to $\Phi_{i_0}(y_0)=y_{i_0}$. 
Continuing inductively, if $G_{j-1}\in \mathcal{H}_{q,q}^*$, then the series composition of $G_{j-1}$ with $G$ is in $\mathcal{H}_{q,y}^*$ since its effective edge interaction is not equal to $1$ as $g(z)=1$ if and only if $z=1$.
Its effective edge interaction is contained in $U$ and therefore not equal to $\infty$.
Next taking the parallel composition with $H_{i_{j}}$ results in the series-parallel graph $G_{j}$ contained in $ \mathcal{H}_{q,y}^*$, as its effective edge interaction is contained in $U$ and therefore not equal to $\infty$; it is not equal to $1$ by construction.

Now we turn to the second case where $y_0\not\in U$ and $y_0\neq 0$.
The algorithm first determines a positive integer $n$ such that $y_0=u_0^n$ for some $u_0\in U$. 
Then it runs the algorithm for the first part on input of $u_0$ and error parameter 
\[\delta=\min\left(\frac{|u_0|}{n-1}, \frac{|u_0|}{e \cdot n|y_0|}\cdot \epsilon\right ),\] 
to obtain a series-parallel graph $H'\in \mathcal{H}_{q,y}^*$ with effective edge interaction $\hat{u} \in B(u_0,\delta)$.
We then output the series-parallel graph $H$ obtained as the $n$-fold parallel composition of $H'$ with itself, which has effective edge interaction $\hat{u}^n$, and we output $Z^\mathrm{dif}(H;q,y)$ computed using \ref{lem:basic Zsame and Zdif}. Clearly $H\in \mathcal{H}_{q,y}^*.$

To prove that the algorithm is also correct in this case, first note that $n=O(|\log(|y_0|)|)$:
the open $U$ contains a set of the form $\{z\in\C \mid -a\leq\arg(z)\leq a, b^{-1}\leq |z|\leq b \}$. 
Now it suffices to take $n=\max(\lceil\pi/a\rceil, \lceil \frac{|\log(|y_0|)|}{\log(b)} \rceil)$. 
We can thus compute $u_0$ in time $\poly(\size(y_0))$. Note that $\size(u_0)=O(n\cdot \size(y_0))$.

% Now we use the first part of this proof with some error bound $\delta$ to get an effective interaction $\hat{u}$ such that $|\hat{u}-u|<\delta$ and that $\hat{u}$ corresponds to a series-parallel gadget with $O(\log(\delta^{-1}))$ edges. 
The output of first procedure, $\hat{u}$, satisfies
%We can bound the difference between $\hat{u}^n$ and $y$, because 
\begin{align*}
|\hat{u}^n-u_0^n|&\leq |\hat{u}-u_0|\cdot n\max(|u_0|,|\hat{u}|)^{n-1}\leq \delta n(|u_0|+\delta)^{n-1}
\\
&\leq \delta n |u_0|^{n-1}e^{(n-1)\delta/|u_0|}=\delta n \frac{|y_0|}{|u_0|}e^{(n-1)\delta/|u_0|}.
\end{align*}
This means that with our choice of $\delta$ we have 
% \[\delta=\min(\frac{|u|}{n-1}, \frac{|u|}{en|y|}\cdot \epsilon),\] 
% because then 
\[
|\hat{u}^n-y_0| \leq \delta n \frac{|y_0|}{|u_0|}e^{(n-1)\delta/|u_0|}\leq \delta n \frac{|y_0|}{|u_0|}e\leq \epsilon.\]

The computation of $\hat{u}$ has a running time of $\poly(\size(u_0,\delta))$, which by construction is $\poly(\size(y_0,\eps))$. 
The series-parallel graph $H$ corresponding to $\hat{u}^n$ is the parallel composition of $n$ copies of the graph corresponding to $\hat{u}$. %It follows by the first part of the proof that $H$ is contained in $\mathcal{H}^*$ and therefore is a gadget.
The number of edges in the gadget is thus $O(n \log(\delta^{-1})) = \poly(\size(\eps,y_0))$.

Finally for $y_0=0$ we simply run the algorithm to find a gadget with an effective edge interaction in $B(\eps/2,\eps/2) \subset B(0,\eps)$, where $\eps/2$ is a non-zero target.
This finishes the proof.
\end{proof}

\section{Box shrinking: a proof of \ref{thm:box shrinking}}\label{sec:complex binary search}
Recall that we consider a linear function $f(y)=Ay+B$, with $A,B$ complex numbers.
Note that in this section we will use $y$ as a variable in this function $f$, and not as a variable in the partition function of the random cluster model.

Our goal is to approximate the root $y^*=-B/A$ of this function. 
We will assume initially that $y^*$ lies within a box of `radius' $D$ (when $A\neq 0$), and every step of the algorithm will shrink this box with a constant factor.
Recall the notation $B_\infty(m,r) = \{z\in\C \mid |\Re(z-m)|<r, |\Im(z-m)|<r \}$ for a square box in the complex plane with radius $r$ and center $m$. 
Also recall that we either have an algorithm at our disposal to find an $0.25$-abs-approximation or to find an $0.25$-arg-approximation to $f$.
We will in this section write $\tilde{f}_\mathrm{abs}(y)$ for any $0.25$-abs-approximation to $f(y)$, and similarly $\tilde{f}_\mathrm{arg}(y)$ for any $0.25$-arg-approximation to $f(y)$.
More concretely, the algorithm takes as input a rational number $y_0\in \Q[i]$ and a rational number $\epsilon>0$, and outputs $\tilde{f}_\mathrm{abs}(\hat{y})$ or respectively $\tilde{f}_{\mathrm{arg}}(\hat{y})$ for some $\hat{y}\in B(y_0,\epsilon)$.

We next describe two variants of our box shrinking procedure, one for the abs-approximator and one for the arg-approximator.
\begin{definition}
\begin{enumerate}[(a)]
    \item Given a square box $B_\infty(m,D)$ with center $m$ and radius $D$. 
With $\epsilon=0.1 D$, we use the algorithm to compute for any $\hat{y}_1\in B(m-\frac{5}{4}D,\epsilon)$, $\hat{y}_2\in B(m+\frac{5}{4}D,\epsilon)$, $\hat{y}_3\in B(m-\frac{5}{4}Di,\epsilon)$ and $\hat{y}_4\in  B(m+\frac{5}{4}Di,\epsilon)$ the values $\tilde{f}_\mathrm{abs}(\hat{y}_1),\ldots, \tilde{f}_\mathrm{abs}(\hat{y}_4)$.
If $\tilde{f}_\mathrm{abs}(\hat{y}_1)\leq \tilde{f}_\mathrm{abs}(\hat{y}_2)$ remove the strip of width $\frac{1}{4}D$ at the right side of the box and at the left side otherwise (see \ref{fig:box_shrinking}).
If $\tilde{f}_\mathrm{abs}(\hat{y}_3)\leq \tilde{f}_\mathrm{abs}(\hat{y}_4)$ remove the strip of width $\frac{1}{4}D$ at the top of the box and and the bottom of the box otherwise.
Denote the resulting box by $\mathcal{S}_\mathrm{abs}(B_\infty(m,D))=B_\infty(m',D')$.
    \item Given a square box $B_\infty(m,D)$ with center $m$ and radius $D$. 
With $\epsilon=0.1D$, we use the algorithm to compute for any $\hat{y}_1\in B(m-\tfrac{5}{4}D,\epsilon)$, $\hat{y}_2\in B(m+\tfrac{5}{4}D,\epsilon)$, $\hat{y}_3\in B(m-\tfrac{5}{4}Di,\epsilon)$ and $\hat{y}_4\in B(m+\tfrac{5}{4}Di,\epsilon)$ the values $\tilde{f}_\mathrm{arg}(\hat{y}_1),\ldots,\tilde{f}_\mathrm{arg}(\hat{y}_4)$. If $\tilde{f}_\mathrm{arg}(\hat{y}_1)-\tilde{f}_\mathrm{arg}(\hat{y}_2)$ is in the interval $(0,\pi)$ we remove the top strip of width $\tfrac{1}{4} D$ of the box, remove the bottom strip otherwise (see~\ref{fig:box_shrinking_arg}).
If $\tilde{f}_\mathrm{arg}(\hat{y}_3)-\tilde{f}_\mathrm{arg}(\hat{y}_4)$ is in the interval $(0,\pi)$, we remove the left strip of width $\tfrac{1}{4} D$, remove the right strip otherwise.
Denote the resulting box by $\mathcal{S}_\mathrm{arg}(B_\infty(m,D))=B_\infty(m',D')$.
\end{enumerate}
\end{definition}

\begin{figure}[ht]
\centering
%\begin{minipage}{0.48\textwidth}
  \begin{tikzpicture}[scale=0.8]
    \fill[fill=red!20] (3,-4) rectangle (4,4);
    \draw (-4,-4) -- (-4,4) -- (4,4) -- (4,-4) -- cycle;
    \draw[<->] (0,2.8) -- node[above]{$D$} (4,2.8);
    \draw[<->] (0,1.2) -- node[above]{$\frac{3}{4}D$} (3,1.2);
    \filldraw (0,0) circle[radius=1pt] node[right]{$m$};
    \filldraw (-5,0) circle[radius=1pt]
      (5,0)  circle[radius=1pt];
    \draw[dashed] (-5,0) circle[radius=0.4]
    (5,0) circle[radius=0.4];
    \node at (-5,-0.9) {$\hat{y}_1$};
    \node at (5,-0.9) {$\hat{y}_2$};
    \draw[<->] (-5,-1.5) -- node[below]{$\frac{5}{2}D$} (5,-1.5);
    \draw[<->] (-5.8,-0.4) -- node[left]{$2\epsilon$} (-5.8,0.4);
    \node at (-3,4.5) {$B_\infty(m,D)$};
  \end{tikzpicture}
  \caption{Box shrinking procedure $\mathcal{S}_\mathrm{abs}$. If $\tilde{f}_\mathrm{abs}(\hat{y}_1) \leq \tilde{f}_\mathrm{abs}(\hat{y}_2)$, the red-shaded strip is removed.}\label{fig:box_shrinking}
%\end{minipage}\hfill 
\end{figure}
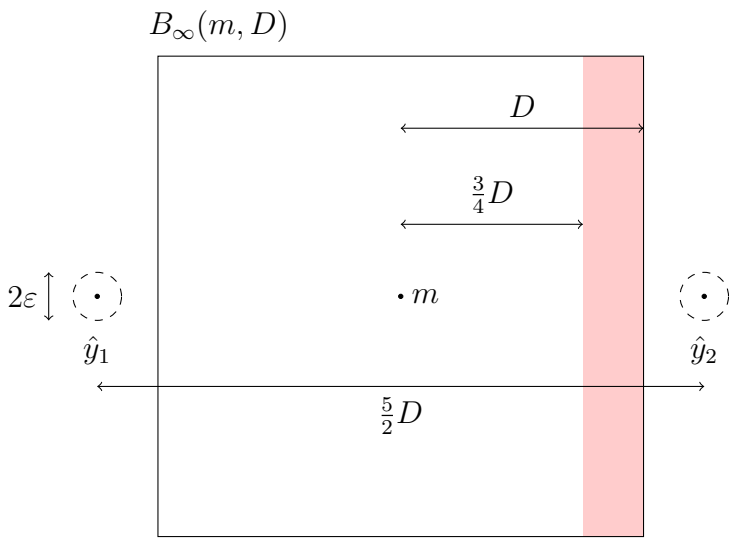
%\begin{minipage}{0.48\textwidth}
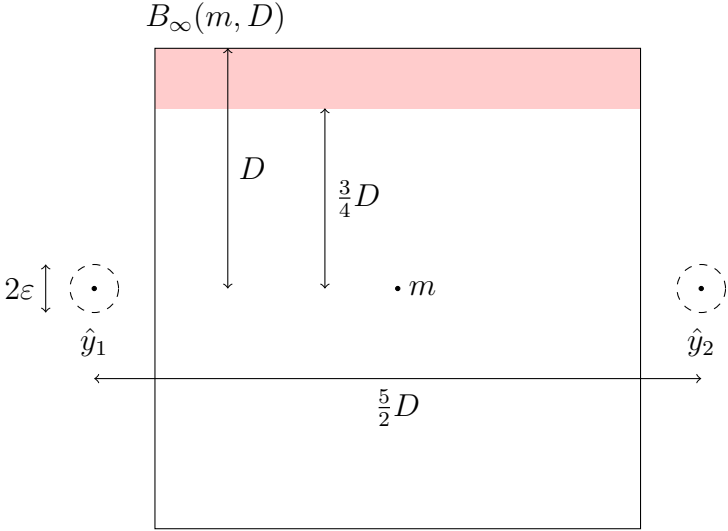
\begin{figure}
\centering
  \begin{tikzpicture}[scale=0.8]
    \fill[fill=red!20] (-4,3) rectangle (4,4);
    \draw (-4,-4) rectangle (4,4);
    \draw[<->] (-2.8,0) -- node[right]{$D$} (-2.8,4);
    \draw[<->] (-5,-1.5) -- node[below]{$\frac{5}{2}D$} (5,-1.5);
    \filldraw (0,0) circle[radius=1pt] node[right]{$m$};
    \filldraw (-5,0) circle[radius=1pt]
      (5,0)  circle[radius=1pt];
    \draw[dashed] (-5,0) circle[radius=0.4]
    (5,0) circle[radius=0.4];
    \node at (-5,-0.9) {$\hat{y}_1$};
    \node at (5,-0.9) {$\hat{y}_2$};
    \draw[<->] (-1.2,0) -- node[right] {$\frac{3}{4}D$} (-1.2,3);
    \draw[<->] (-5.8,-0.4) -- node[left]{$2\epsilon$} (-5.8,0.4);
    \node at (-3,4.5) {$B_\infty(m,D)$};
  \end{tikzpicture}
  \caption{Box shrinking procedure $\mathcal{S}_\mathrm{arg}$. If $\tilde{f}_\mathrm{arg}(\hat{y}_1)-\tilde{f}_\mathrm{arg}(\hat{y}_2)$ is in the interval $(0,\pi)$, the red-shaded strip is removed.}\label{fig:box_shrinking_arg}
%\end{minipage}
\end{figure}

In the next two lemmas we record important observations about this procedure. The first says we have a good understanding of the midpoint and radius of the resulting box, and the second says we are guaranteed to retain the zero $y^*$ in our box.
\begin{lemma}\label{lem:basic observations box shrinkin}
The resulting radius $D'$ satisfies $D'=\tfrac{7}{8}D$ and the resulting center $m'$ satisfies $m'=m+(\tfrac{\pm 1}{8}+\tfrac{\pm i}{8}) D$.
\end{lemma}

\begin{lemma}\label{lem:box shrinking does not throw away y*}
Suppose that $A\neq 0$ and $\mathcal{S} \in \{\mathcal{S}_\mathrm{abs},\mathcal{S}_\mathrm{arg}\}$. If $y^*\in B_\infty(m,D)$, then $y^*\in \mathcal{S}(B_\infty(m,D))$.
\end{lemma}

\begin{proof}
We prove the cases $\mathcal{S}=\mathcal{S}_\mathrm{abs}$ and $\mathcal{S}=\mathcal{S}_\mathrm{arg}$ separately.
\indent {\bf Case 1}[$\mathcal{S}=\mathcal{S}_\mathrm{abs}$]
\\
In this case we write $\tilde{f}$ for $\tilde{f}_\mathrm{abs}$. Suppose that $y^*\not\in \mathcal{S}(B_\infty(m,D))$, we will reach a contradiction and thereby prove the lemma. We may assume wlog that $\tilde{f}(\hat{y}_1)\leq \tilde{f}(\hat{y}_2)$, but that $y^*$ is in the strip of width $\frac{1}{4}D$ at the right side of the box. Note that $\hat{y}_1$ and $\hat{y}_2$ are certainly not equal to $y^*$.

Then we see that
\[
    \frac{\tilde{f}(\hat{y}_1)}{\tilde{f}(\hat{y}_2)}  > e^{-0.5} \frac{|A\hat{y}_1+B|}{|A\hat{y}_2+B|} = e^{-0.5} \frac{|\hat{y}_1 - y^*| }{ |\hat{y}_2 - y^*|}.
\]
We will now proceed by bounding $|\hat{y}_1-y^*|^2-|\hat{y}_2-y^*|^2$ and $|\hat{y}_2-y^*|$. For this denote $y_1=m-\frac{5}{4}D$ and $y_2=m+\frac{5}{4}D$. See~\ref{fig:box_shrinking} for a sketch of the situation.

We easily find the upper bounds $|y_1-y^*| \leq \frac{\sqrt{97}}{4}D$ and $|y_2-y^*|\leq\frac{\sqrt{5}}{2}D$ (the maxima are reached when $y^*$ is in a corner of the strip), so $|\hat{y}_2-y^*| \leq \frac{\sqrt{5}}{2}D + \epsilon < \tfrac{3}{2} D$, using that $\epsilon=0.1D$.

Next we see, using Pythagoras and the fact that $y_1$ and $y_2$ have the same imaginary part, that $|y_1-y^*|^2-|y_2-y^*|^2 $ is equal to
\begin{align*}
&\Re(y_1-y^*)^2 - \Re(y_2-y^*)^2  + \Im(y_1-y^*)^2 - \Im(y_2-y^*)^2  \\
=&(\Re(y_1-y^*) - \Re(y_2-y^*)) \cdot (\Re(y_1-y^*) + \Re(y_2-y^*)) \\
%&+ (\Im(y_1-y^*) - \Im(y_2-y^*)) \cdot (\Im(y_1-y^*) + \Im(y_2-y^*))\\
\geq & \tfrac{5}{2}D \cdot \tfrac{3}{2}D + 0\cdot 2D = \tfrac{15}{4}\cdot D^2.
\end{align*} 
Including the uncertainty in $\hat{y}_i$ versus $y_i$ yields that $ |\hat{y}_1-y^*|^2 - |\hat{y}_2-y^*|^2 $ is lower bounded by
\begin{align*}
   &(|y_1-y^*|-\epsilon)^2-(|y_2-y^*|+\epsilon)^2 \\
    &=|y_1-y^*|^2-|y_2-y^*|^2-2\epsilon(|y_1-y^*|+|y_2-y^*|)\\
    &\geq \tfrac{15}{4}D^2-0.2D\left(\tfrac{\sqrt{97}}{4}D+\tfrac{\sqrt{5}}{2}D\right) > 3D^2.
\end{align*}

%The uncertainty in $\hat{y}_i$ versus $y_i$ means that any term $y_i-y^*$ in this bound can change at most $\epsilon$ to obtain $\hat{y}_i-y^*$.
%These added $\epsilon$'s lead to the bound
%\begin{align*}
%    |\hat{y}_1-y^*|^2 - |\hat{y}_2-y^*|^2 & \geq \left(\tfrac{5}{2}D - 2\epsilon \right) \left( \tfrac{3}{2}D - 2\epsilon \right) - 2\epsilon \left(2D+2\epsilon\right) = 2.55D^2.
%\end{align*}
Putting this together yields
\[
\frac{|\hat{y}_1-y^*|^2}{|\hat{y}_2-y^*|^2}=1+\frac{|\hat{y}_1-y^*|^2 - |\hat{y}_2-y^*|^2}{|\hat{y}_2-y^*|^2}> 3 >e,
\]
which proves that $\tilde{f}(\hat{y}_1)/\tilde{f}(\hat{y}_2)>1$. This contradicts $\tilde{f}(\hat{y}_1)\leq \tilde{f}(\hat{y}_2)$, completing the proof of Case 1.

{\bf Case 2 }[$\mathcal{S}=\mathcal{S}_\mathrm{arg}$]
\\
In this case we write $\tilde{f}$ for $\tilde{f}_\mathrm{arg}$.
Suppose that $y^*\not\in\mathcal{S}(B_\infty(m,D))$, we will again reach a contradiction. This time we may assume that $\tilde{f}(\hat{y}_1)-\tilde{f}(\hat{y}_2)$ is in the interval $(0,\pi)$, but that $y^*$ is in the strip of width $\frac{1}{4}D$ at the top of the box.

We setup some notation
\begin{align*}
    d_1 &= |\hat{y}_1-y^*|, &   d_2 &= |\hat{y}_2-y^*|, & \alpha &=\arg(\hat{y}_2-y^*) - \arg(\hat{y}_1-y^*),\\
      x &= d_1^2-d_2^2, &         y &= 2d_1d_2, &              z &= d_1^2+d_2^2,\\
    y_1 &= m-\tfrac{5}{4}D, &   y_2 &= m+\tfrac{5}{4}D, &      w &= y\sin(\alpha).
\end{align*}
Note that $\arg(f(\hat{y}_2))-\arg(f(\hat{y}_1))$ is equal to
\begin{align*}
& \arg(A\hat{y}_2+B)-\arg(A\hat{y}_1+B)\\
=&\arg(A\hat y+B-(Ay^*+B))-\arg(A\hat{y}_1+B-(Ay^*+B))=\alpha,
\end{align*}
and we can interpret this geometrically as the angle between the segments $\hat{y}_1$ to $y^*$ and $\hat{y}_2$ to $y^*$, measured clockwise from $\hat{y}_2$ to $\hat{y}_1$.
By definition of the arg-approximation, the difference $\tilde{f}(\hat{y}_2)-\tilde{f}(\hat{y}_1)$ will differ at most $0.5<\frac{\pi}{6}$ from $\alpha$.
Also observe that $x,y,z$ satisfy the relation $x^2+y^2=z^2$.
The cosine rule in the triangle formed by $\hat{y}_1,\hat{y}_2,y^*$ yields $z=y\cos(\alpha)+|\hat{y}_1-\hat{y}_2|^2$. Together this implies the following relation:
\begin{align*}
    x^2&=z^2-y^2\\
    &= \left(y\cos(\alpha)+|\hat{y}_1-\hat{y}_2|^2\right)^2-y^2\\
    &= -y^2\sin^2(\alpha)+2y\sin(\alpha)\cdot |\hat{y}_1-\hat{y}_2|^2\cot(\alpha) + |\hat{y}_1-\hat{y}_2|^4 \\
    &= -w^2+2w\cdot |\hat{y}_1-\hat{y}_2|^2\cot(\alpha) + |\hat{y}_1-\hat{y}_2|^4,
    \end{align*}
and hence
\begin{equation}\label{eq:cot}
    \cot(\alpha) = \frac{x^2+w^2-|\hat{y}_1-\hat{y}_2|^4}{2w\cdot |\hat{y}_1-\hat{y}_2|^2}.
\end{equation}
The distance $|y_1-y_2|$ is exactly $\tfrac{5}{2}D$, so we see that
\[
2.3D=\tfrac{5}{2}D-2\epsilon<|\hat{y}_1-\hat{y}_2| < \tfrac{5}{2}D+2\epsilon = 2.7D.
\]
%$|\hat{y}_1-\hat{y}_2|$ is between $\tfrac{5}{2}D-2\epsilon$ and $\tfrac{5}{2}D+2\epsilon$.
We can also bound $|y_i-y^*|\leq \frac{\sqrt{97}}{4}D$ for both $i=1,2$.
Now we continue to bound $x$ and $w$. Using Pythagoras, we compute that $|y_1-y^*|^2-|y_2-y^*|^2$ is equal to
\begin{align*}
& \Re(y_1-y^*)^2 - \Re(y_2-y^*)^2  + \Im(y_1-y^*)^2 - \Im(y_2-y^*)^2  \\
&=(\Re(y_1-y^*) - \Re(y_2-y^*)) \cdot (\Re(y_1-y^*) + \Re(y_2-y^*)) \\
%&+ (\Im(y_1-y^*) - \Im(y_2-y^*)) \cdot (\Im(y_1-y^*) + \Im(y_2-y^*))\\
&=\tfrac{5}{2}D\cdot (\Re(y_1-y^*) + \Re(y_2-y^*)).
\end{align*}
The last factor has absolute value at most $2D$, so in total the absolute value is at most $5D^2$.
With the perturbations $\hat{y}_i$ from $y_i$ this gives the bound
\begin{align*}
|x|&=\left|d_1^2-d_2^2 \right| \\
&\leq \max\left\{\left|(|y_1-y^*|+\epsilon)^2-(|y_2-y^*|-\epsilon)^2\right| \right., 
\\
&\quad\quad\quad\quad\quad \left. \left|(|y_1-y^*|-\epsilon)^2-(|y_2-y^*|+\epsilon)^2\right|\right\}\\
&= \left| |y_1-y^*|^2 - |y_2-y^*|^2 \right| + 2\epsilon(|y_1-y^*|+|y_2-y^*|)\\
&\leq 5D^2 + 0.2D\cdot 2\tfrac{\sqrt{97}}{4}D < 6D^2.
\end{align*}

We note that $w$ is four times the area of the triangle formed by the points $\hat{y}_1,\hat{y}_2,y^*$.
To obtain a lower bound, we may assume that $\hat{y}_1=y_1+(\epsilon+\epsilon i)$, because this point is closer the the line through $y^*$ and $\hat{y}_2$ then any point in $B(y_1,\epsilon)$. Similarly we may assume that $\hat{y}_2=y_2+(-\epsilon+\epsilon i)$, and then the minimal area is attained when $\Im(y^*)=\Im(m)+\tfrac{3}{4}D$.
For the upper bound, we assume that $\hat{y}_1=y_1+(-\epsilon-\epsilon i)$ and $\hat{y}_2=y_2+(\epsilon-\epsilon i)$, and the maximum is when $\Im(y^*)=\Im(m)+D$.
This yields the bounds
\[
 2.99D^2 = 2(\tfrac{5}{2}D-2\epsilon)(\tfrac{3}{4}D-\epsilon) \leq w \leq 2(\tfrac{5}{2}D+2\epsilon)(D+\epsilon) = 5.94D^2.
\]

Plugging in all these bounds into~\ref{eq:cot} yields
\[
-\sqrt{3}<-1.40 < \cot(\alpha) < 1.37<\sqrt{3}.
\]
We see that $\arg(\hat{y}_1-y^*)\in (\pi,\tfrac{3}{2}\pi)$ and $\arg(\hat{y}_2-y^*) \in (\tfrac{3}{2}\pi,2\pi)$, so that $\alpha \in (0,\pi)$. Then the bounds on $\cot(\alpha)$ imply $\alpha \in \left(\frac{1}{6}\pi,\frac{5}{6}\pi\right)$.
Including the approximation error, we find that $\tilde{f}(\hat{y}_2)-\tilde{f}(\hat{y}_1)$ is in the interval $(0,\pi)$, contradicting the assumption that the opposite difference $\tilde{f}(\hat{y}_1)-\tilde{f}(\hat{y}_2)$ is in the interval $(0,\pi)$. This completes the proof.\qed
\end{proof}

We now use the box shrinking procedure to prove~\ref{thm:box shrinking}, which we restate for convenience of the reader.
\begin{namedtheorem*}{Theorem 3.4}
Let $A,B$ be complex numbers and let $C>0$ be a rational number such that 
$|A|$ and $|B|$ are both at most $C$, and both are either $0$ or at least $1/C$.
Assume one of the following: 
\begin{itemize}
\item there exists a $\poly(\size(y_0,\eps))$-time algorithm to compute on input of $y_0\in\Q[i]$ and a rational number $\epsilon>0$ an $0.25$-abs-approximation of $A\hat{y}+B$ for some algebraic number $\hat{y} \in B(y_0,\eps)$, or,
\item there exists a $\poly(\size(y_0,\eps))$-time algorithm to compute on input of $y_0\in\Q[i]$ and a rational number $\epsilon>0$ an $0.25$-arg-approximation of $A\hat{y}+B$ for some algebraic number $\hat{y}\in B(y_0,\epsilon)$.
\end{itemize}
Then there exists an algorithm that on input of a rational $\delta>0$ and $C>0$ as above that outputs ``$A=0$'' when $A=0$ and $B\neq 0$, and that outputs ``$A\neq 0$'' and a number $\bar{y}\in \mathbb{Q}[i]$ such that $-B/A\in B_\infty(\bar{y},\delta/2)$ when $A\neq 0$. 
When $A=B=0$ it is allowed to output anything. 
The running time is $\poly(\size(C,\delta))$.
\end{namedtheorem*}

\begin{proof}
We will first decide whether $A=0$ or $A\neq 0$.

If we are using the abs-approximation algorithm we do the following:
compute $\tilde{f}_\mathrm{abs}(\hat{y})$ for any $\hat{y}$ such that $|\hat{y}|>5C^2$.
If $\tilde{f}_\mathrm{abs}(\hat{y})< 2C$, we output ``$A=0$'' and terminate, else we output ``$A\neq 0$'' and continue with the rest of the algorithm.

Instead if we are using the arg-approximation algorithm:
with $\epsilon=0.1C^2$, we take $\hat{y}_1 \in B(-5C^2,\epsilon)$, $\hat{y}_2 \in B(5C^2,\epsilon)$ and compute $\tilde{f}_\mathrm{arg}(\hat{y}_1)$ and $\tilde{f}_\mathrm{arg}(\hat{y}_2)$.
If $\tilde{f}_\mathrm{arg}(\hat{y}_2)-\tilde{f}_\mathrm{arg}(\hat{y}_1)$ is in the interval $(-\frac{1}{4}\pi,\frac{1}{4}\pi)$, we output ``$A=0$'' and terminate, else we output ``$A\neq 0$'' and continue with the rest of the algorithm.

Next, the idea is to apply the box shrinking procedure, with $\mathcal{S}=\mathcal{S}_\mathrm{abs}$ or $\mathcal{S}=\mathcal{S}_\mathrm{arg}$ depending on the algorithm at our disposal.
We take as starting box the square box with center $0$ and `radius' $D=C^2$.
We apply the box shrinking procedure 
\[
n\coloneqq \left\lceil \log(2D/\delta)/\log(8/7) \right\rceil
%n\coloneqq 1+\max\{\log(D/2\delta)/\log(8/7),\log\left(\frac{1+\eta}{1-\eta}C^2D\right)/\log(8/7)\}
\]
many times.
Finally the algorithm outputs the center $m_n$ of the ball $B_\infty(m_n,D_n)=\mathcal{S}^{\circ n}(B_\infty(0,D))$.

We first argue correctness of the algorithm and deal with the running time after that.
We may assume that not both $A$ and $B$ are equal to $0$, otherwise the algorithm is allowed to output anything anyway.

If $A=0$ and $B\neq 0$ we have $f(y)=B\neq 0$ for any $y$. For the abs-approximator we see that $\tilde{f}_\mathrm{abs}(\hat{y})\leq e^{0.25}|B|<2C$ and then the algorithm indeed outputs ``$A=0$''.
On the other hand, for the arg-approximator we have that $\tilde{f}_\mathrm{arg}(\hat{y}_2)-\tilde{f}_\mathrm{arg}(\hat{y}_1)$ is at most $2\cdot0.25<\frac{\pi}{4}$ away from $\arg(f(\hat{y}_2))-\arg(f(\hat{y}_1))=0$, so again the algorithm correctly outputs ``$A=0$''.

If on the other hand $A\neq 0$, we see for the abs-approximator that $|f(\hat{y})|\geq |A||\hat{y}|-|B|>C^{-1}\cdot 5C^2 -C=4C$. This is in particular non-zero, so $\tilde{f}(\hat{y})\geq e^{-0.25}|f(\hat{y})|>2C$ and indeed the algorithm outputs ``$A\neq 0$''.
For the arg-approximator, we again adopt the notation from the proof of \ref{lem:box shrinking does not throw away y*}, and we will prove that $\cos(\alpha)<0$. First we see that $|\hat{y}_1-\hat{y}_2|^2 \geq (10D-2\epsilon)^2=96.04D^2$. We can bound $d_i^2$ by using Pythagoras, to see that $d_i^2\leq (6D+\epsilon)^2+(D+\epsilon)^2=38.42D^2$, yielding $z \leq 76.84D^2$. Then
\[
y\cos(\alpha)=z-\left| \hat{y}_1 - \hat{y}_2 \right|^2\leq 76.84D^2-96.04D^2<0.
\]
This means that $\alpha\in (\tfrac{1}{2}\pi,\tfrac{3}{2}\pi)$, and $\tilde{f}(\hat{y}_2) - \tilde{f}(\hat{y}_1)$ is in the interval $(\tfrac{1}{3}\pi,\tfrac{5}{3}\pi)$, since $2\cdot0.25<\pi/6$. Then the algorithm indeed outputs ``$A\neq0$''.

When $A\neq 0$, we also see that $y^*\in B_\infty(0,D)$ by our choice of $D$. 
By \ref{lem:box shrinking does not throw away y*} we then know that $y^*\in B_\infty(m_n,D_n)$.
Because $D_n=(7/8)^nD \leq \delta/2$, this box is small enough.
We thus conclude that our algorithm is correct and move on to the analysis of the running time.

The running time of the algorithm is dominated by the applications of the box shrinking.
First of all we note that $n=O(\log(C/\delta))$.
The smallest $\epsilon$ that we encounter is $0.05\delta$.
By \ref{lem:basic observations box shrinkin} and induction, it follows that after $k$ steps, the diagonal of the current box, $D_k$, is of the form $(7/8)^k D$ and its center, $m_k$, is of the form $0$ plus a rational multiple (of size $O(k)$) of $D$ and hence is itself rational. 
The values of $y_i$ that we encounter are of the form $m_k \pm \tfrac{5}{4}D_k$ and $m_k \pm \tfrac{5}{4}D_ki$.
Therefore the $y_i$ that are used as input for the assumed algorithm have their sizes bounded by $O(\log(D))+O(n)=O(\log(C/\delta))$.
Hence we obtain a running time of
\[
O\big(\log(C/\delta)\big)\poly(\log(C/\delta))=\poly(\log(C/\delta)).\qed
\]
\end{proof}

\section{Concluding remarks}\label{sec:conclusion}

\paragraph{Planar graphs: inside the disk $|q-1|<1$} An interesting question left open by our results is whether approximately computing the chromatic polynomial of planar graphs is \textsc{\#P}-hard for all non-real algebraic $q$. We refer to \ref{fig:chromatic hardness region} for a figure displaying the region for which we can prove hardness with the aid of a computer;
here we use the computer to try and verify the condition in \ref{thm:general main result}.
The family $\mathcal{H}_{q,y}^*$, appearing in this condition, is now restricted to series-parallel graphs. 
We can allow this family to contain planar, two-terminal graphs where the terminal are on the same face. Looking at wheel graphs this can be seen to improve \ref{fig:chromatic hardness region}, but is not yet enough to resolve the question completely.

\paragraph{Real evaluations of the chromatic polynomial for planar graphs}
The focus in the present paper has been on non-real evaluations. 
Using our techniques we can also obtain results for real evaluations of the chromatic polynomial.
It suffices for $q\in \mathbb{R}$ to be able to get density on the real line with effective edge interactions of planar graphs (where, as above, the two terminals are on the same face). 
Having this, the machinery of the present paper can be adapted in a straightforward way to prove hardness of approximating the absolute value at $q$. 
To obtain density of the effective edge interactions at $q$ it suffices to find an effective interaction that is finite and less than $-1$. %proof in in two separate case: 32/27<q<2 and $q>2$ where $1-q<-1$!
When restricting to the family $\mathcal{H}^*_{q,0}$, this is possible if and only if $q\in (32/27,2)$ by~\cite{bencs2022location}. 
When we replace the family $\mathcal{H}^*_{q,0}$ by the family of two-terminal planar graphs with the two terminals on the same face, one can get a negative effective interaction, thus density in $\mathbb{R}$, for $q$ inside the union of three intervals $(2,3)\cup (3,t_1)\cup (t_2,4)$, where $t_1\approx 3.618032$ and $t_2\approx 3.618356$ by results of \cite{thomassen1997zero} and \cite{PerrettThomassen18}.
Consequently, it is then \textsc{\#P}-hard to approximate the absolute value of the chromatic polynomial for planar graphs for any algebraic $q$ in any of these intervals. 

Interestingly, there is the value $\tau+2\in (t_1,t_2)$ (where $\tau$ is the golden ratio) at which the chromatic polynomial of any planar graphs is positive by a result of \cite{Tuttegolden}, see also~\cite{PerrettThomassen18}. 
This suggests that the computational complexity of approximating the absolute value of the chromatic polynomial at $\tau+2$ is an intriguing problem.\footnote{In a recent seminar talk (see \url{https://homepages.dcc.ufmg.br/~gabriel/AGT/wp-content/uploads/2021/02/30_Gordon_Royle.pdf}) it was announced that Gordon Royle and Melissa Lee proved that $t_1$ can be chosen to be $\tau+2$.}

Finally, we ask about the complexity of approximating the absolute value of the chromatic polynomial at large values of $q$ on planar graphs (computing the sign is hard for general graphs~\cite{GoldbergJerrumsign}).
\cite{BL} (see also \cite{Woodalllargest}) showed that planar graphs have no chromatic roots larger than $5$ indicating that the constructions that we employ in the present paper are not possible. This suggests that determining the complexity of approximating the chromatic polynomial for planar graphs in this regime is an interesting problem.

\paragraph{Almost bounded degree (planar) graphs} 
Combining our techniques with some ingredients from~\cite{bencs2022location} and some tools from complex dynamics we expect that for $\Delta\geq 3$ and non-real algebraic $q$ such that $1<|q-1|< \Delta-1$,  approximating the chromatic polynomial for planar graphs of maximum degree at most $\Delta$ with one vertex of potentially unbounded degree is \textsc{\#P}-hard. 
We leave the details for follow up work\footnote{Meanwhile this has been confirmed in \cite{SdR}.}.
This relates to a result of \cite{GSVinapprox}, who showed that for even positive integer $q$, corresponding to proper $q$-colorings, proved that it is \textsc{NP}-hard to approximate the evaluation of the chromatic polynomial at $q$ on all graphs of maximum degree $\Delta$ when $q<\Delta$.
This should be contrasted with a result from~\cite{PatelRegts} which shows that for all graphs of maximum degree at most $\Delta$ and any $q$ such that $|q|>6.91\Delta$ there exists an efficient algorithm to approximate the chromatic polynomial.
This algorithm is based on the Barvinok's interpolation method~\cite{barbook} and a zero-freeness result for the chromatic polynomial for bounded degree graphs~\cite{FP08,JPS13}.
The zero-free region can be extended to the family graphs of maximum degree at most $\Delta$ where one vertex may have unbounded degree at the cost of replacing $6.91\Delta$ by $7.97\Delta+1$ using~\cite[Corollary 6.4]{Sokalbounded}. 
Using Sokal's representation of the chromatic polynomial of a graph with one vertex of potentially unbounded degree as an evaluation of the partition function of a (multivariate) random cluster model with external fields of a bounded degree graph~\cite{Sokalbounded}, the algorithm from~\cite{PatelRegts} can be adapted to run in polynomial time for this class of graphs as well.

\paragraph{The reliability polynomial}

Recall that $T(G;x,y)$ denotes the Tutte polynomial of a graph $G$.
For a connected graph $G$ and $x=1$ we define
\begin{align*}
C(G;y) & \coloneqq (y-1)^{|V|-1}T(G;1,y)
=\sum_{\substack{ A\subseteq E \\ (V,A) \text{ connected}}}(y-1)^{|A|} 
\\
&= \lim_{q\to0} \tfrac{1}{q}Z(G;q,y).
\end{align*}
This is up to a transformation the reliability polynomial, i.e.\@ $(1-p)^{|E|}C(G;\frac{1}{1-p})$ gives the probability that the graph $G$ remains connected if edges are independently selected with probability $p$, and deleted with probability $1-p$ (see e.g. \cite{Sokalsurvey}). 
Clearly, the approach for proving density in the present paper does not apply directly. 
However, a variation of our methods can be applied in this setting. 
We will leave this for future work.

\begin{acknowledge}
This research was funded by the Netherlands Organisation of Scientific Research (NWO): VI.Vidi.193.068. We thank the anonymous referees for constructive feedback.
\end{acknowledge}
\bibliography{chromatic}%
%\bibliographystyle{plain}
%\bibliography{chromatic}
\end{document}